\documentclass[journal]{IEEEtran}
\usepackage{amsmath,amssymb,amsthm}
\usepackage{graphicx}
\usepackage{booktabs}
\usepackage{float}
\usepackage{cite}
\usepackage{url}
\usepackage{array}
\newtheorem{theorem}{Theorem}
\newtheorem{lemma}{Lemma}
\newtheorem{corollary}{Corollary}
\newtheorem{proposition}{Proposition}
\newfloat{algorithm}{tbp}{loa}
\floatname{algorithm}{Algorithm}
\newenvironment{algorithmic}[1][]{\begin{list}{}{\setlength{\leftmargin}{1em}\setlength{\itemsep}{2pt}}}{\end{list}}
\newcommand{\Require}{\item[]\textbf{Input:} }
\newcommand{\State}{\item[]}
\newcommand{\For}[1]{\item[]\textbf{For} #1:}
\newcommand{\EndFor}{\item[]\textbf{End for}}
\newcommand{\While}[1]{\item[]\textbf{While} #1:}
\newcommand{\EndWhile}{\item[]\textbf{End while}}
\newcommand{\If}[1]{\item[]\textbf{If} #1:}
\newcommand{\EndIf}{\textbf{End if}}
\newcommand{\Return}{\textbf{return} }
\newcommand{\Comment}[1]{\quad\textit{(#1)}}

\begin{document}

\title{Sample-Based Quantum Krylov Galerkin Solvers for Elliptic Partial Differential Equations on Curved Geometries with Hardware Validation}

\author{Doyeol~(David)~Ahn,~\IEEEmembership{Fellow,~IEEE,}
        and~Byeongyong~Park
\thanks{D. Ahn is with the Center for Quantum Information Processing, Department of Electrical and Computer Engineering, University of Seoul, Seoul 02504, Republic of Korea, and also with Singularity Quantum Inc., Newport Beach, CA 92660 USA (e-mail: dahn@uos.ac.kr).}
\thanks{B. Park is with the Center for Quantum Information Processing, University of Seoul, Seoul 02504, Republic of Korea.}}

\markboth{}{}

\maketitle

\begin{abstract}
Sample-based quantum solvers measure time-evolved Krylov states and solve the problem classically on the observed basis states. We develop this approach for the positive-definite linear systems of elliptic partial differential equations. A Galerkin solve on the sampled Fourier modes gives the solution. We derive energy-norm error bounds that connect the measured outcome law to the solution error, a sample complexity, and a support-completeness theorem. The theorem shows that time averaging recovers solution modes that are hidden from every selection based on the source. For curved cross sections we give a constructive procedure: conformal coordinates fix the matrix class, isometries of the cross section give postselection rules that keep every solution mode, and the Lie derivative of the metric bounds the error from geometric uncertainty. We apply the method to pulsatile flow in elliptic, lens-shaped, and curved ducts and to non-separable problems with up to 16 qubits. On 16 qubits the sampler reaches 0.2\% error with 3\% of the Fourier modes, while a selection from the source stalls at 13\% error. At these sizes a classical residual-greedy selection is more accurate, but its cost grows with the full dimension, whereas the sampler and the reduced Galerkin solve scale with the number of sampled modes. On IonQ and IBM processors with six and eight qubits, the measured outcome laws follow a replacement-noise model with a common decay of about 0.6\% per two-qubit gate, and a parity rule derived from the geometry detects noisy shots at the predicted rate.
\end{abstract}


\section{Introduction}\label{sec:introduction}
Many problems in science and engineering lead to large linear systems $Ax=b$. Quantum algorithms for this task began with the algorithm of Harrow, Hassidim, and Lloyd~\cite{hhl} and were improved by later methods based on linear combinations of unitaries and quantum signal processing~\cite{cks}. These algorithms prepare the solution as a quantum state. In order to read out the solution, however, one must estimate many amplitudes, and the required circuits are deep. Variational methods~\cite{vqls,xu_vqls} use shallow circuits, but they rely on a classical optimization without guarantees.

Sample-based methods take a different route. In sample-based quantum diagonalization (SQD), a quantum processor generates bitstrings, and a classical computer diagonalizes the Hamiltonian in the subspace spanned by these bitstrings~\cite{sqd}. Its performance depends on recovery procedures that repair noisy samples by using constraints visible in the measurement outcomes, such as particle-number symmetry~\cite{sqd}. These procedures have been extended to strongly correlated systems~\cite{csqd} and, through code-space recovery, to problems without a native constraint~\cite{csr}. Sample-based Krylov quantum diagonalization (SKQD) uses time-evolved states $e^{-iHT}|\psi\rangle$ as the source of the bitstrings and comes with convergence guarantees~\cite{skqd}. The same idea applies to linear systems: one measures quantum Krylov states $e^{-iAT}|b\rangle$ generated by the Hermitian matrix $A$ itself and solves the system classically on the observed basis states. A related approach was developed concurrently by Park \emph{et al.}~\cite{park}, who apply sample-based Krylov diagonalization to general Hermitian linear systems with a variational reconstruction based on the cost Hamiltonian and prove support-discovery guarantees for that setting. The present paper is self-contained. We treat positive-definite systems from elliptic partial differential equations and reconstruct the solution by a Galerkin solve. The energy-norm bounds of Sec.~\ref{sec:galerkin-theory} are derived directly from the outcome law and do not rely on the sampling theorems of Ref.~\cite{park}.

The present work builds on SKQD~\cite{skqd}, but its target is different. In SKQD the target is an eigenvector, and the convergence analysis relies on the concentration of the ground state. For a linear system the target $x^*=A^{-1}b$ is not an eigenvector, and three questions arise that do not occur for eigenvalue problems. First, a solution mode can be absent from every single Krylov state by cancellation, and it is not clear that sampling finds it. Theorem~\ref{thm:support} shows that the dephased outcome law always covers the support of the solution, with an explicit floor on the sampling probability. Second, the error must be measured for the solution vector, not for an energy. Theorem~\ref{thm:error} and its corollaries connect the outcome law directly to the energy-norm error of the reconstruction. Third, the short time step that is natural for eigenvalue SKQD fails for linear systems on fine grids, as Supplementary Lemma~\ref{lem:finite} and Supplementary Table~\ref{tab:times} show. The optimality of the Galerkin solution is the classical C\'ea property, and we use it as such. The new elements are the bounds that link this property to the quantum outcome law.

In this paper we develop the method for linear systems that arise from elliptic partial differential equations. These systems are positive definite, and their solutions are smooth. Two observations guide the work. First, the solution of such a system is not concentrated in the grid-point basis, but it is strongly concentrated in a sine or cosine basis matched to the boundary conditions. We therefore measure in such a basis. Second, for a positive-definite system a Galerkin solve on the sampled set is simpler than a variational reconstruction and is optimal in the energy norm. We derive error bounds for this reconstruction and relate them to the sampling probabilities.

Elliptic problems on curved domains are naturally formulated on a Riemannian manifold~\cite{Nakahara2003,Jost2017}. A cross section with coordinates $u^i$ carries a metric $g_{ij}$, and the physical operator becomes a weighted Laplace--Beltrami operator. In this language a change of coordinates is a diffeomorphism, and the weak form of the problem is invariant under it. The metric then enters the discrete matrix only through the combinations $\sqrt g\,g^{ij}$ and $\sqrt g$. In two dimensions the first combination is invariant under conformal changes of the metric, which makes the stiffness matrix independent of the geometry. Symmetries of the cross section are isometries generated by Killing fields $X$ with $\mathcal L_Xg=0$, where $\mathcal L_X$ is the Lie derivative. They give selection rules that tell us which measurement outcomes can be discarded. Deformations of the cross section are generated by vector fields $V$, and the Lie derivative $\mathcal L_Vg$ gives their first-order effect on the matrix. We use these notions throughout Sec.~\ref{sec:geometry} and summarize them in Supplementary Sec.~\ref{app:riemann}.

We then formulate the method for elliptic problems on a Riemannian cross section. We use the formulation for the implicit step of pulsatile flow in elliptic, lens-shaped, and curved ducts. These ducts are models of compressed and eccentrically narrowed coronary arteries~\cite{Zamir2005}. We compare the quantum sampler with classical mode selection, and we test the outcome laws on six- and eight-qubit IonQ and IBM processors. The theory and the examples together identify where the sampling is efficient, and they give the conditions under which a benefit over classical mode selection can be expected.

The contributions of this paper are as follows. (1) We formulate Krylov-state sampling in any fixed basis with an efficient transform (Sec.~\ref{sec:basis}). This follows from unitary equivalence, but it is what makes the method useful for partial differential equations. (2) For positive-definite systems we replace the variational reconstruction by a Galerkin solve and derive expected and high-probability energy-norm bounds, a sample complexity, a sampling efficiency, a computable residual certificate, and a support-completeness theorem for the dephased outcome law (Sec.~\ref{sec:galerkin-theory}). The optimality of the Galerkin solution is the classical C\'ea property; the new elements are the bounds that connect the energy-norm error to the outcome law. (3) We give a constructive procedure for curved cross sections (Sec.~\ref{sec:geometry}). The individual facts that we use, namely the conformal invariance of the Dirichlet energy, Killing symmetries, and first-order shape calculus, are classical, and we do not claim them as new. The contribution is their use for the quantum sampler: they fix the measurement basis and the matrix class, give postselection filters that provably keep every solution mode, and give an error budget for geometric uncertainty. (4) We test the method on elliptic, lens-shaped, and curved ducts in simulation and on two quantum processors (Secs.~\ref{sec:results} and~\ref{sec:hardware}), and on two non-separable problems with up to 16 qubits in simulation (Sec.~\ref{sec:large}). The hardware tests use six-qubit circuits on both processors and eight-qubit circuits on the IBM processor. They verify the measurement convention, the replacement-noise model with the same decay rate per two-qubit gate on both platforms, the geometric parity rule as an error-detection mechanism, and the end-to-end Galerkin reconstruction from hardware samples. (5) We compare the sampler with classical mode selection both at equal subspace size and in cost (Sec.~\ref{sec:cost}). The classical selections that win at equal size require vectors of the full dimension $n=2^q$, whereas the sampler and the reduced Galerkin solve do not. From these comparisons we state four conditions that a problem must satisfy before quantum Krylov-state sampling can offer a benefit (Sec.~\ref{sec:conclusion}).

The duct examples serve as test problems. Exact reference solutions are available for them, so that every bound can be checked mode by mode. In the Fourier basis matched to the boundary conditions, however, their matrices are nearly diagonal, and a classical selection of modes is easy. These examples therefore test the theory and the hardware, and at these sizes they are not intended to demonstrate an advantage. In order to test the method in a larger and less favorable setting, Sec.~\ref{sec:large} adds two non-separable examples with up to 16 qubits: a curved lens-shaped vessel and a multiscale vorticity field in a curved channel. In these examples the sampled set is a small fraction of the space.

The remainder of this paper is organized as follows. Section~\ref{sec:algorithm} describes the sampling protocol, measurement in a general basis, and the Galerkin reconstruction. Section~\ref{sec:galerkin-theory} states the energy-norm guarantees and the support-completeness theorem. Section~\ref{sec:geometry} formulates the method on Riemannian cross sections and introduces the duct examples. Section~\ref{sec:results} presents numerical results, and Sec.~\ref{sec:hardware} describes the hardware validation on two platforms. The proofs, the effect of finite time averaging, supporting numerical tests, and further details are given in the Supplementary Material. Section~\ref{sec:discussion} discusses the cost of the classical selections and the regime in which the sampler can be useful, and Sec.~\ref{sec:conclusion} concludes.

\section{Sample-based Krylov solver}\label{sec:algorithm}
\subsection{Problem and sampling protocol}\label{sec:problem}
Let $A=A^\dagger\in\mathbb C^{N\times N}$ be invertible, with $N=2^n$ and $\|b\|=1$. We assume spectral bounds
\begin{equation}
 \operatorname{spec}(A)\subseteq[-B,-a]\cup[a,B],\qquad 0<a\leq B,\qquad \kappa=B/a .
 \label{eq:spectrum}
\end{equation}
The bounds may be conservative. The normalized solution is
\begin{equation}
 |x\rangle=\frac{A^{-1}|b\rangle}{\|A^{-1}b\|},
 \label{eq:target}
\end{equation}
and we write $x_z=\langle z|x\rangle$ for its amplitude at the basis string $z$.

In quantum Krylov-state sampling, each shot prepares $|b\rangle$, evolves it under $A$ for a time $T$, and measures $e^{-iAT}|b\rangle$ in the computational basis. If $T$ is drawn from a distribution $\mu$, the outcome law is
\begin{equation}
 q(z)=\mathbb E_{T\sim\mu}\,|\langle z|e^{-iAT}|b\rangle|^2 .
 \label{eq:outcome-law}
\end{equation}
In this paper $\mu$ is a finite set of times $T_1,\dots,T_d$ with equal or prescribed shot counts. The evolution can be compiled by product formulas or by randomized methods such as qDRIFT~\cite{qdrift}. Let $\mathcal S$ be the set of distinct observed strings and $\mathcal S_M$ the set observed in $M$ independent shots.

\subsection{Measurement in a general orthonormal basis}\label{sec:basis}
The computational basis is not essential to the protocol. Let $F$ be a unitary matrix with an efficient circuit $F_q$, and define the basis vectors $|z\rangle_F=F^\dagger|z\rangle$. Measuring $F_qe^{-iAT}|b\rangle$ in the computational basis is equivalent to measuring $e^{-iAT}|b\rangle$ in the basis $\{|z\rangle_F\}$. Since
\begin{equation}
 F e^{-iAT}F^\dagger=e^{-i\tilde AT},\qquad \tilde A=FAF^\dagger,\qquad \tilde b=Fb,
 \label{eq:basis-change}
\end{equation}
the protocol in the basis $F$ is the computational-basis protocol for the unitarily equivalent system $\tilde A\tilde x=\tilde b$. The spectrum, the spectral bounds of Eq.~\eqref{eq:spectrum}, and $\|A^{-1}b\|$ are unchanged. Therefore every statement about the computational-basis protocol holds in the basis $F$, with $x_z$ replaced by the coefficients $\langle z|_F x\rangle$. We state the results of Sec.~\ref{sec:galerkin-theory} directly in the basis $F$. The basis matters because concentration is basis dependent. For elliptic problems with smooth data, a sine or cosine basis matched to the boundary conditions concentrates the solution on a few low wave vectors, whereas the computational (grid-point) basis does not. Sections~\ref{sec:geometry} and~\ref{sec:results} use this freedom.

\subsection{Positive-definite systems and Galerkin reconstruction}\label{sec:galerkin}
When $A$ is Hermitian positive definite, a second reconstruction is available. The duct examples use real symmetric matrices. Let $x^*=A^{-1}b$ denote the unnormalized solution, with $\|b\|=1$, and let $\hat x_z=\langle z|_Fx^*\rangle$ be its coefficients in the measurement basis. For a set $\mathcal S$ we define $V_{\mathcal S}=\mathrm{span}\{|z\rangle_F:z\in\mathcal S\}$. The Galerkin solution $x_{\mathcal S}\in V_{\mathcal S}$ satisfies $\langle v,b-Ax_{\mathcal S}\rangle=0$ for all $v\in V_{\mathcal S}$. With $\hat A=FAF^\dagger$ and $\hat b=Fb$, it is obtained from
\begin{equation}
 \hat A[\mathcal S,\mathcal S]\,y=\hat b[\mathcal S],\qquad x_{\mathcal S}=F^\dagger\mathrm{Embed}(y,\mathcal S),
 \label{eq:galerkin}
\end{equation}
where $\hat A[\mathcal S,\mathcal S]$ is the principal submatrix indexed by $\mathcal S$. Since it is a principal submatrix of a positive-definite matrix, Eq.~\eqref{eq:galerkin} has a unique solution.

\begin{lemma}[Galerkin optimality]\label{lem:galerkin}
Let $\|v\|_A^2=\langle v,Av\rangle$. The Galerkin solution satisfies $\|x_{\mathcal S}-x^*\|_A=\min_{v\in V_{\mathcal S}}\|v-x^*\|_A$. If $\mathcal S\subseteq\mathcal S'$, then $\|x_{\mathcal S'}-x^*\|_A\leq\|x_{\mathcal S}-x^*\|_A$.
\end{lemma}

Lemma~\ref{lem:galerkin} is the classical C\'ea property~\cite{saad}. For general Hermitian systems one can instead minimize a projected cost function on $V_{\mathcal S}$, as in variational formulations~\cite{cost_hamiltonian,xu_vqls} and in the solver of Ref.~\cite{park}. For positive-definite systems the Galerkin solution is simpler: it minimizes the energy-norm error, requires only the principal submatrix $\hat A[\mathcal S,\mathcal S]$, and its error does not increase on nested sets.

\subsection{A practical solver for positive-definite systems}\label{sec:practical}
Algorithm~\ref{alg:solver} summarizes the variant used in Secs.~\ref{sec:geometry} and~\ref{sec:results}. Phase~1 measures Krylov states in the basis $F$ with a classical Galerkin solve after each state. Phase~2 is an optional classical expansion that adds the indices with the largest residual components. The function \textsc{Filter} removes indices at which the solution vanishes by symmetry or constraint, such as a parity rule or a zero mode. A filter must retain every potentially nonzero solution coefficient; discarding physically admissible modes invalidates the coverage bounds. The sampling times may be taken from a uniform grid, whose outcome law approaches the dephased law of Theorem~\ref{thm:support} at the rate given by Supplementary Lemma~\ref{lem:finite}, or from the heuristic schedule of Algorithm~\ref{alg:times}. The latter is motivated in Sec.~\ref{sec:times} but carries no separate guarantee.

\begin{algorithm}[t]
\caption{Sample-based Krylov solver for positive-definite systems}\label{alg:solver}
\begin{algorithmic}[1]
\Require $A$, $b$, basis $F$ with circuit $F_q$, sampling times $T_1,\dots,T_d$, shots $M$, tolerance $\varepsilon$, expansion size $m$, maximum size $S_{\max}$
\State $\hat b\gets Fb$;\quad $\mathcal S\gets\emptyset$
\For{$k=1,\dots,d$}\Comment{Phase 1: quantum sampling}
  \State prepare $e^{-iAT_k}|b\rangle$, apply $F_q$, and measure $M/d$ times; let $\mathcal B$ be the observed indices
  \State $\mathcal S\gets\mathcal S\cup\textsc{Filter}(\mathcal B)$
  \State solve $\hat A[\mathcal S,\mathcal S]\,y=\hat b[\mathcal S]$;\quad $r\gets\hat b-\hat A\,\mathrm{Embed}(y,\mathcal S)$
  \If{$\|r\|\leq\varepsilon$} \Return $F^\dagger\mathrm{Embed}(y,\mathcal S)$ \EndIf
\EndFor
\While{$\|r\|>\varepsilon$ and $|\mathcal S|<S_{\max}$}\Comment{Phase 2: classical expansion}
  \State add to $\mathcal S$ the $m$ admissible indices outside $\mathcal S$ with the largest $|r_z|$
  \State solve $\hat A[\mathcal S,\mathcal S]\,y=\hat b[\mathcal S]$;\quad $r\gets\hat b-\hat A\,\mathrm{Embed}(y,\mathcal S)$
\EndWhile
\State \Return $F^\dagger\mathrm{Embed}(y,\mathcal S)$
\end{algorithmic}
\end{algorithm}

\begin{algorithm}[t]
\caption{Heuristic multiscale sampling times}\label{alg:times}
\begin{algorithmic}[1]
\Require positive estimates $\lambda_{\min}$, $\lambda_{\max}$ (certified bounds $a$, $B$ if guarantees are required), number of states $d\geq3$
\State $t_{\mathrm{low}}\gets\pi/\lambda_{\max}$;\quad $t_{\mathrm{high}}\gets2\pi/\lambda_{\min}$;\quad $\mathcal T\gets\{0\}$
\For{$k=1,\dots,d-1$}
  \State $\mathcal T\gets\mathcal T\cup\{t_{\mathrm{low}}(t_{\mathrm{high}}/t_{\mathrm{low}})^{(k-1)/(d-2)}\}$
\EndFor
\State \Return $\mathcal T$
\end{algorithmic}
\end{algorithm}

\section{Energy-norm guarantees for Galerkin reconstruction}\label{sec:galerkin-theory}
For the positive-definite case of Sec.~\ref{sec:galerkin}, an expectation bound in the energy norm follows directly from Lemma~\ref{lem:galerkin}. The proofs of the results of this section and of Sec.~\ref{sec:geometry} are given in Supplementary Secs.~\ref{sup:proofs3} and~\ref{sup:proofs4}. Throughout this section the measurement basis is $\{|z\rangle_F\}$, $x^*=A^{-1}b$ with $\|b\|=1$, $\hat x_z=\langle z|_Fx^*\rangle$, and $\lambda_{\min}$, $\lambda_{\max}$ are the extreme eigenvalues of $A$. For sampling times $T_1,\dots,T_d$ with $M/d$ shots each, the time-averaged outcome law is
\begin{equation}
 q(z)=\frac1d\sum_{k=1}^{d}\bigl|\langle z|_Fe^{-iAT_k}|b\rangle\bigr|^2 .
 \label{eq:q}
\end{equation}
Let $P_\lambda$ be the spectral projector of $A$ onto the eigenvalue $\lambda$, and define $a_\lambda(z)=\langle z|_FP_\lambda b\rangle$. Then
\begin{equation}
 \hat x_z=\sum_\lambda\frac{a_\lambda(z)}{\lambda},\qquad \hat b_z=\sum_\lambda a_\lambda(z).
 \label{eq:a}
\end{equation}

\begin{theorem}[Expected energy-norm error]\label{thm:error}
Let $\mathcal S$ be the set of distinct indices observed with $M/d$ shots on each Krylov state. Then
\begin{equation}
 \mathbb E\bigl[\|x_{\mathcal S}-x^*\|_A^2\bigr]\leq\lambda_{\max}\sum_z|\hat x_z|^2e^{-Mq(z)}.
 \label{eq:bound}
\end{equation}
\end{theorem}

For unequal shot counts $M_k$ on the $k$th state, the same proof gives the bound
\[
 \lambda_{\max}\sum_z|\hat x_z|^2\exp\Bigl[-\sum_kM_kp_k(z)\Bigr].
\]
 Theorem~\ref{thm:error} is an expectation statement. If $R(M)$ denotes its right-hand side divided by $\|x^*\|_A^2$, Markov's inequality gives
\[
 \Pr\bigl(\|x_{\mathcal S}-x^*\|_A>\varepsilon\|x^*\|_A\bigr)\leq R(M)/\varepsilon^2 .
\] A mean-square bound is therefore not a high-probability guarantee unless the budget includes the failure probability. Corollary~\ref{cor:highprob-galerkin} below gives such a budget.

\begin{corollary}[Sample complexity]\label{cor:complexity}
Let $\varepsilon>0$, and let $\mathcal E_\varepsilon$ be a set of indices with
\[
 \lambda_{\max}\sum_{z\notin\mathcal E_\varepsilon}|\hat x_z|^2\leq\frac{\varepsilon^2}{2}\|x^*\|_A^2 .
\]
Define $q_{\min}(\mathcal E_\varepsilon)=\min_{z\in\mathcal E_\varepsilon}q(z)$. If
\begin{equation}
 M\geq\frac{\ln(2\kappa/\varepsilon^2)}{q_{\min}(\mathcal E_\varepsilon)},
 \label{eq:M}
\end{equation}
then $\mathbb E[\|x_{\mathcal S}-x^*\|_A^2]\leq\varepsilon^2\|x^*\|_A^2$, where $\kappa=\lambda_{\max}/\lambda_{\min}$.
\end{corollary}

\begin{corollary}[High-probability Galerkin reconstruction]\label{cor:highprob-galerkin}
Let $\varepsilon$, $\mathcal E_\varepsilon$, and $q_{\min}(\mathcal E_\varepsilon)$ be as in Corollary~\ref{cor:complexity}, but with the stronger tail condition $\lambda_{\max}\sum_{z\notin\mathcal E_\varepsilon}|\hat x_z|^2\leq\varepsilon^2\|x^*\|_A^2$. For $0<\delta<1$, if
\begin{equation}
 M\geq\frac{\ln(|\mathcal E_\varepsilon|/\delta)}{q_{\min}(\mathcal E_\varepsilon)},
 \label{eq:galerkin-highprob-budget}
\end{equation}
then with probability at least $1-\delta$ every index of $\mathcal E_\varepsilon$ is observed and $\|x_{\mathcal S}-x^*\|_A\leq\varepsilon\|x^*\|_A$.
\end{corollary}
Compared with Corollary~\ref{cor:complexity}, the high-probability budget replaces $\ln(2\kappa/\varepsilon^2)$ by $\ln(|\mathcal E_\varepsilon|/\delta)$. The essential set of Eq.~\eqref{eq:galerkin-highprob-budget} is defined in the energy norm, so it can contain modes of very small normalized weight. Complete discovery of $\mathcal E_\varepsilon$ is a sufficient condition and can be pessimistic, because a missed mode of small weight costs little error. Any filter must preserve the solution support, and pruning after sampling must be controlled separately.

\begin{proposition}[Lower bound and sampling efficiency]\label{prop:eff}
Let $K_\varepsilon$ be the smallest $|\mathcal S|$ for which the Galerkin error is at most $\varepsilon\|x^*\|_A$. Any sampling protocol needs $M\geq K_\varepsilon$ shots, because each shot adds at most one index. Since $q_{\min}(\mathcal E_\varepsilon)\leq1/|\mathcal E_\varepsilon|$, the right-hand side of Eq.~\eqref{eq:M} is at least $|\mathcal E_\varepsilon|\ln(2\kappa/\varepsilon^2)$.
\end{proposition}

We therefore define the sampling efficiency
\begin{equation}
 \eta_{\mathrm s}=|\mathcal E_\varepsilon|\,q_{\min}(\mathcal E_\varepsilon)\leq1 .
 \label{eq:eta}
\end{equation}
A small value of $\eta_{\mathrm s}$ identifies an essential index with a low sampling probability. Inefficiency can arise from repeated samples of already observed indices or from probability mass outside $\mathcal E_\varepsilon$. This is a coupon-collector problem with a nonuniform distribution~\cite{Flajolet1992}. The relevant property of $q$ is its flatness on $\mathcal E_\varepsilon$, not its resemblance to $|\hat x|^2$. The two bounds of Proposition~\ref{prop:eff} need not match when $\eta_{\mathrm s}$ is small.

\begin{corollary}[Replacement noise]\label{cor:noise}
Suppose that each shot is independently replaced, with probability $p<1$, by an index drawn from an arbitrary distribution. Then the outcome law satisfies $q'(z)\geq(1-p)q(z)$, and the sufficient budget of Corollary~\ref{cor:complexity} increases at most by the factor $1/(1-p)$.
\end{corollary}

By Lemma~\ref{lem:galerkin}, adding corrupted indices to a retained set cannot increase its error. At a fixed total number of shots, however, replacement can remove useful ideal samples and increase the error relative to an uncontaminated run. The shot bound assumes independent replacement. Arbitrary gate and readout channels need not satisfy the pointwise dominance $q'\geq(1-p)q$. For them we use a weaker statement. If the actual law $q'$ lies within total-variation distance $\eta$ of $q$, then $|q'(z)-q(z)|\leq\eta$ for every $z$, and Corollaries~\ref{cor:complexity} and~\ref{cor:highprob-galerkin} hold with $q_{\min}$ replaced by $q_{\min}-\eta$ whenever this is positive.

For the residual $r=b-Ax_{\mathcal S}$, the identity $\|x^*-x_{\mathcal S}\|_A^2=r^\dagger A^{-1}r$ gives the computable certificate
\begin{equation}
 \frac{\|x^*-x_{\mathcal S}\|_A}{\|x^*\|_A}\leq\sqrt\kappa\,\frac{\|r\|}{\|b\|}.
 \label{eq:residual-certificate}
\end{equation}
Consequently, the residual stopping rule of Algorithm~\ref{alg:solver} certifies a relative energy error of at most $\sqrt\kappa\,\varepsilon$, not generally $\varepsilon$. For an energy tolerance $\varepsilon_A$ one uses a certified condition-number bound and the residual tolerance $\varepsilon_A/\sqrt\kappa$. Approximate matrix assembly and inexact reduced solves require additional error terms.

\subsection{Support completeness and dephasing}\label{sec:support}
The following statement explains why a single state can fail and a long time average cannot, and it gives an a priori floor for $q_{\min}$. We define the dephased law $q^\infty(z)=\sum_\lambda|a_\lambda(z)|^2$, where the sum runs over the distinct eigenvalues of $A$.

\begin{theorem}[Support completeness]\label{thm:support}
For every index $z$,
\begin{equation}
 q^\infty(z)\geq\frac{|\hat x_z|^2}{\sum_\lambda\lambda^{-2}}\geq\frac{|\hat x_z|^2}{\operatorname{tr}A^{-2}}.
 \label{eq:support}
\end{equation}
In particular, $\operatorname{supp}\hat x\subseteq\operatorname{supp}q^\infty$.
\end{theorem}

The probability of $z$ in the single state $|b\rangle$ is $|\hat b_z|^2=|\sum_\lambda a_\lambda(z)|^2$. This sum can vanish by cancellation while $\hat x_z=\sum_\lambda a_\lambda(z)/\lambda$ does not, because the weights $1/\lambda$ break the cancellation. A classical selection of modes by the magnitude of $\hat b$ has the same deficiency. Section~\ref{sec:theorycheck} shows that this happens for the elliptic duct. For the dephased law, Eq.~\eqref{eq:support} gives
\begin{equation}
 q^\infty_{\min}(\mathcal E_\varepsilon)\geq\frac{\min_{z\in\mathcal E_\varepsilon}|\hat x_z|^2}{\operatorname{tr}A^{-2}},
 \label{eq:qmin-floor}
\end{equation}
so that Corollaries~\ref{cor:complexity} and~\ref{cor:highprob-galerkin} become a priori budgets that depend only on the solution profile and the spectrum. The floor depends on the dimension through $\operatorname{tr}A^{-2}$, which is dominated by the lowest eigenvalues for elliptic operators. For $A=aI$ the bound is exact.

\begin{corollary}[Constant coefficient]\label{cor:const}
If every basis vector $|z\rangle_F$ is an eigenvector of $A$, then $q(z)=q^\infty(z)=|\hat b_z|^2$ for all sampling times, and $\operatorname{supp}\hat x=\operatorname{supp}\hat b$.
\end{corollary}

For uniform times the outcome law approaches $q^\infty$ only when the eigenvalue differences that carry the solution are dephased over the time window. A single uniform grid can fail completely by aliasing, so that a mode hidden by cancellation in $\hat b$ stays hidden (Supplementary Lemma~\ref{lem:finite}). Incommensurate or varied time steps avoid this failure.

\subsection{A multiscale heuristic}\label{sec:times}
A long uniform step aliases the high eigenvalues, while a short step does not dephase the low ones. A mixture of complete schedules has a provable coverage property (Supplementary Lemma~\ref{lem:hedging}). Algorithm~\ref{alg:times} instead selects isolated, logarithmically spaced times, and we use it as an empirical heuristic. Its lower time scale $\pi/\lambda_{\max}$ avoids aliasing of the highest eigenvalues, and its upper scale $2\pi/\lambda_{\min}$ resolves the lowest eigenvalues, which dominate $\hat x$ through the weights $1/\lambda$ of Eq.~\eqref{eq:a}.

\section{Geometric formulation}\label{sec:geometry}
In this section we formulate the solver for elliptic problems on a general Riemannian cross section. The individual geometric results of this section are classical. We state them in the form that the sampler needs and use them as a design procedure, which selects the measurement basis, the postselection filter, and the error budget. We first state the geometric method as a sequence of steps. We then show that the guarantees of Sec.~\ref{sec:galerkin-theory} apply to the resulting matrices and classify the approximations that a realization requires. An interpretation of the sampling in terms of Green's functions is given in Supplementary Sec.~\ref{app:supp}. Elliptic, lens-shaped, and curved ducts appear at the end of the section as concrete examples.

\subsection{The geometric method}\label{sec:geometric-method}
Let $\Omega$ be a cross section, and let $\Phi:R\to\Omega$ be a coordinate map from a simple coordinate domain $R$, such as a rectangle, a strip, or a quarter rectangle. The map induces a metric $g_{ij}$ on $R$. The geometric method consists of five steps.
\begin{enumerate}
\item \emph{Coordinates.} Choose $\Phi$ such that the wall of $\Omega$ is mapped to coordinate lines of $R$. Conformal maps are preferred in two dimensions (Proposition~\ref{prop:conformal}), and orthogonal tube coordinates in curved vessels.
\item \emph{Weak form.} Write the physical problem in the invariant weak form~\eqref{eq:weak}. The metric enters only through $\sqrt g$ and $g^{ij}$.
\item \emph{Reference operator and basis.} Split the discrete matrix as $A=A_0+V$, where $A_0$ is a reference operator whose eigenbasis $F$ has an efficient quantum circuit. For a rectangle in conformal coordinates with a constant stiffness coefficient, $A_0$ is the flat stiffness matrix and $F$ is a sine or cosine transform. All mode coupling is then contained in $V$.
\item \emph{Sampling and reconstruction.} Measure Krylov states in the basis $F$ (Sec.~\ref{sec:basis}) with a schedule of Sec.~\ref{sec:support} or~\ref{sec:times}, apply a filter that encodes the symmetries of $\Omega$, and reconstruct by the Galerkin solve of Eq.~\eqref{eq:galerkin}.
\item \emph{Approximation.} Replace the exact metric, domain, or circuit by approximations of Class I or Class II (Sec.~\ref{sec:hierarchy}).
\end{enumerate}
Steps 1--3 define a discrete model, which already includes discretization and any coordinate truncation. Step~5 separates further model and sampler approximations. Physical accuracy requires control of all these contributions.

\subsection{Elliptic problems on a Riemannian cross section}
Let $\Omega$ have local coordinates $(u^1,u^2)$ and metric $g_{ij}$, and let $\sqrt g=\sqrt{\det g_{ij}}$. We consider the weak problem: find $w$ with $w=0$ on the wall such that
\begin{equation}
 \int_R\bigl(\varrho\,g^{ij}\partial_iw\,\partial_jv+\gamma\,wv\bigr)\sqrt g\,du^1du^2=\int_Rf\,v\,\sqrt g\,du^1du^2
 \label{eq:weak}
\end{equation}
for all admissible $v$. Here $\varrho>0$ and $\gamma\geq0$ are coefficient fields. Equation~\eqref{eq:weak} is the weak form of $-\nabla_g\cdot(\varrho\nabla_gw)+\gamma w=f$, where $\nabla_g$ denotes the Levi-Civita gradient and divergence~\cite{Nakahara2003}. A conservative discretization on a coordinate grid gives
\begin{equation}
\begin{aligned}
 A&=K+C,\qquad K=\nabla_h^{T}W\nabla_h,\\
 W&=\operatorname{blockdiag}\bigl(\varrho\sqrt g\,[g^{ij}]\bigr),\\
 C&=\operatorname{diag}(\gamma\sqrt g),\qquad b=\operatorname{diag}(\sqrt g)f.
\end{aligned}
 \label{eq:KC}
\end{equation}
where $\nabla_h$ is the discrete gradient. For an orthogonal coordinate system $[g^{ij}]$ is diagonal and $K$ is a symmetric five-point stencil.

\begin{proposition}[Structure]\label{prop:structure}
The matrix $A$ of Eq.~\eqref{eq:KC} is real symmetric positive definite whenever $\varrho>0$, $\gamma\geq0$, and the wall carries a Dirichlet condition. Therefore the results of Sec.~\ref{sec:galerkin-theory} apply to every problem of the form~\eqref{eq:weak}, for any metric and any orthogonal measurement basis.
\end{proposition}

The implicit Euler step of the unsteady Stokes equation, Eq.~\eqref{eq:step} below, is the case $\varrho=1$, $\gamma=\Gamma$, and $f=\Gamma w^n+G^{n+1}$.

\subsection{Conformal reduction in two dimensions}
\begin{proposition}[Conformal reduction]\label{prop:conformal}
If the coordinates are conformal, $g_{ij}=h^2\delta_{ij}$, then $\sqrt g\,g^{ij}=\delta^{ij}$. As a result, $K$ is the flat stiffness matrix with coefficient $\varrho$, independent of the metric, and the geometry enters only through $C=\operatorname{diag}(\gamma h^2)$ and $b=\operatorname{diag}(h^2)f$.
\end{proposition}

Proposition~\ref{prop:conformal} is the discrete form of the conformal invariance of the Dirichlet energy in two dimensions. When the coordinate domain is a rectangle and $\varrho=1$, a sine or cosine basis diagonalizes $K$ exactly, and all mode coupling comes from the metric factor $h^2$ through $C$. A proper simply connected planar domain admits a conformal map to the disk. Mapping a physical boundary to a rectangle or a strip requires appropriate boundary treatment, and corners, narrow features, and mapping singularities can affect numerical conditioning. Patient-specific cross sections segmented from images can be treated by a numerical conformal map~\cite{Driscoll2002,Marshall2007}, which replaces the analytic $h^2$ of the examples below by a computed one.

\subsection{Symmetries and the Lie derivative}\label{sec:symmetry}
Step~4 of the geometric method applies a filter that removes indices at which the solution vanishes. The Lie derivative gives a systematic way to find such indices. A vector field $X$ on $\Omega$ is a Killing field if $\mathcal L_Xg=0$, where $\mathcal L_X$ denotes the Lie derivative (Supplementary Sec.~\ref{app:riemann}). Its flow is then a one-parameter group of isometries. A reflection that maps $\Omega$ onto itself and preserves $g$ is a discrete isometry.

\begin{proposition}[Symmetry selection rule]\label{prop:symmetry}
(i) Let $X$ be a Killing field that is tangent to the wall and has a complete flow on $\Omega$, and assume $\mathcal L_X\varrho=\mathcal L_X\gamma=0$. Then $\mathcal L_X$ commutes with the solution operator of Eq.~\eqref{eq:weak}. In particular, if $\mathcal L_Xf=0$, then $\mathcal L_Xw=0$, and if $-i\mathcal L_Xf=mf$, then $-i\mathcal L_Xw=mw$.

(ii) Let $P$ be a unitary matrix with $PA=AP$ and $Pb=\chi b$. Then every Krylov state $e^{-iAT}b$ and the solution $x^*$ are eigenvectors of $P$ with eigenvalue $\chi$. If the measurement basis $F$ diagonalizes $P$, with $FPF^\dagger=\operatorname{diag}(\chi_z)$, then $q(z)=0$ and $\hat x_z=0$ for every index with $\chi_z\neq\chi$.
\end{proposition}

Part (ii) is the discrete form of part (i). It requires a grid that is mapped onto itself by the isometry, so that $P$ is a permutation of grid points. Part (ii) also shows that a filter that keeps the indices with $\chi_z=\chi$ satisfies the requirement of Algorithm~\ref{alg:solver}: it removes only shots that the ideal sampler never produces and keeps every coefficient that can be nonzero. The parity filter of the lens (Sec.~\ref{sec:ducts}) is this rule for the two reflections $\sigma\to2\pi-\sigma$ and $\tau\to-\tau$. For the disk, the Killing field $\partial_\vartheta$ restricts a rotationally symmetric load to the modes with zero angular momentum. Conversely, if a coefficient field is not invariant, the selection rule is lost. In a curved vessel the coefficient $\varrho=H$ is not invariant under one of the reflections, and Sec.~\ref{sec:geomresults} shows that the number of nonzero modes doubles.

\subsection{Curved and twisted vessels}\label{sec:curved}
For a vessel whose centerline has curvature $\kappa_c$ and torsion $\tau_c$, we use the orthogonal tube coordinates of Germano~\cite{Germano1982}. In these coordinates the metric is orthogonal, and the effect of torsion on the axial flow appears only at second order~\cite{Germano1982}. Let $H=1-\kappa_c\rho\cos\vartheta$ be the metric factor along the centerline, where $(\rho,\vartheta)$ are polar coordinates in the cross section and $H>0$ is required throughout the tube. In the Stokes limit, with the secondary flow neglected, we use the reduced scalar model of the form~\eqref{eq:weak} with
\begin{equation}
 \varrho=H,\qquad \gamma=\Gamma H+\frac{\kappa_c^2}{H}.
 \label{eq:curved}
\end{equation}
The secondary (Dean) flow~\cite{Dean1927} enters through the convective terms. In an implicit--explicit scheme these terms are moved to the right-hand side, so that the matrix remains of the form~\eqref{eq:KC}. Unlike the conformal case, the curvature enters the stiffness matrix through $\varrho=H$. By Proposition~\ref{prop:structure}, the theory applies to this positive-definite reduced model. Its reaction term and load conventions require independent validation against the full vector Stokes equations before quantitative curved-vessel predictions are claimed. Torsion and secondary-flow effects are not validated by the tests of Sec.~\ref{sec:geomresults}.

\subsection{Approximation hierarchy}\label{sec:hierarchy}
A realization on hardware requires approximations. We divide them into two classes.

\paragraph{Class I: approximations of the operator} These replace $A$ by a nearby matrix $\tilde A$. Examples are the error of a numerical conformal map, the truncation of an infinite coordinate strip, the neglect of higher orders in the curvature, and the penalization of a domain on a bounding box. They change the solution that the algorithm computes.

\begin{proposition}[Metric perturbation]\label{prop:perturb}
Let $\tilde A=A+\delta A$ with $|\langle v,\delta A\,v\rangle|\leq\varepsilon\langle v,Av\rangle$ for all $v$ and some $\varepsilon<1$, and let $\tilde b=b+\delta b$ and $\tilde x=\tilde A^{-1}\tilde b$. Then
\begin{equation}
 \frac{\|\tilde x-x^*\|_A}{\|x^*\|_A}\leq\frac{\varepsilon+\|\delta b\|_{A^{-1}}/\|x^*\|_A}{1-\varepsilon}.
 \label{eq:perturb}
\end{equation}
In particular, a relative error $|\delta(\sqrt g)|\leq\varepsilon\sqrt g$ of the metric factor in $C$, or $|\delta\varrho|\leq\varepsilon\varrho$ in $K$, satisfies the assumption on $\delta A$.
\end{proposition}

Geometric uncertainty, such as an error in the segmented wall of a vessel, is a Class I approximation of a particular kind. We describe it by a vector field $V$ that deforms the cross section. Its flow $\psi_s$ maps the reference cross section to the deformed one. By pulling the problem back to the reference coordinates, we keep the grid and the basis $F$ and change only the metric, $g\to\psi_s^*g=g+s\,\mathcal L_Vg+O(s^2)$.

\begin{proposition}[First-order shape perturbation]\label{prop:shape}
In two dimensions, let $\mu_1\geq\mu_2$ be the eigenvalues of $g^{-1}\mathcal L_Vg$ at each point. To first order in $s$, the deformation changes the coefficients of Eq.~\eqref{eq:KC} by
\begin{equation}
 \delta\bigl(\sqrt g\,g^{ij}\bigr)=-s\sqrt g\,\bigl(\mathcal L_Vg\bigr)^{ij}_{\circ},\qquad
 \delta\sqrt g=s\sqrt g\,\operatorname{div}_gV,
 \label{eq:shape}
\end{equation}
where $(\mathcal L_Vg)_\circ=\mathcal L_Vg-(\operatorname{div}_gV)\,g$ is the trace-free part and indices are raised with $g$. Hence Proposition~\ref{prop:perturb} applies, to first order, with
\begin{equation}
 \varepsilon=|s|\max\Bigl\{\sup_\Omega\tfrac12(\mu_1-\mu_2),\ \sup_\Omega\tfrac12|\mu_1+\mu_2|\Bigr\},
 \label{eq:shape-eps}
\end{equation}
and with $\delta b=s\operatorname{diag}(\operatorname{div}_gV)\,b$, when $\varrho$, $\gamma$, and $f$ are transported with the deformation. A conformal deformation, $\mathcal L_Vg=2\sigma g$, leaves $K$ unchanged and changes only $C$ and $b$.
\end{proposition}

The stiffness matrix thus responds only to the shear part of the deformation, and the reaction term and the load respond only to the change of area. This is the infinitesimal form of Proposition~\ref{prop:conformal}.

A pointwise relative metric bound does not imply the same relative bound on a signed or cancelling load, so the load term must be estimated separately. A Class I approximation that is itself derived from a weak form of the type~\eqref{eq:weak}, with $\varrho>0$ and $\gamma\geq0$, preserves positive definiteness and therefore Lemma~\ref{lem:galerkin}. This is analogous to the conserving approximations of many-body theory, which are derived from a functional and therefore preserve the conservation laws~\cite{Baym1961}.

\paragraph{Class II: approximations of the sampler} These change only the states that are measured. Examples are qDRIFT compilation~\cite{qdrift}, the Trotter splitting of Eq.~\eqref{eq:trotter}, truncated Walsh series of diagonal phases, approximate state preparation, gate errors, and readout errors. They do not change $A$ or $b$.

\begin{corollary}[Sampler approximations]\label{cor:classII}
Class II approximations leave the Galerkin solution on a given set unchanged. They affect the outcome law, and Theorem~\ref{thm:error} and Corollaries~\ref{cor:complexity} and~\ref{cor:highprob-galerkin} hold with the law that the approximate sampler actually produces.
\end{corollary}

Circuit approximations can nevertheless destroy essential support. If an essential mode has zero probability, no finite shot budget recovers it without a separate classical expansion. At a fixed budget, both the achieved accuracy and the classical subspace cost can worsen. The hierarchy gives a simple design rule. The accuracy budget is spent on Class I, where Proposition~\ref{prop:perturb} applies. Class II approximations cost shots and coverage, and their effect enters only through the outcome law.

\subsection{Examples: elliptic, lens-shaped, and curved ducts}\label{sec:ducts}
We now apply the geometric method to three cases. The first two are exact conformal instances. The third uses the curvature terms of Sec.~\ref{sec:curved} and is treated numerically in Sec.~\ref{sec:geomresults}.

\subsubsection{The implicit step in conformal coordinates}
We consider fully developed axial flow $w(x,y,t)$ in a straight duct driven by a pressure gradient $G(t)$. Let $(u,v)$ be conformal coordinates of the cross section with scale factor $h(u,v)$, so that $\nabla^2=h^{-2}(\partial_{uu}+\partial_{vv})$. We nondimensionalize lengths by a geometric length $\ell_g$, the velocity by $G\ell_g^2/\mu$, and time by $\ell_g^2/\nu$. The implicit Euler step of the unsteady Stokes equation then reads
\begin{equation}
 \begin{gathered}-(\partial_{uu}+\partial_{vv})\,w^{n+1}+\Gamma h^2w^{n+1}=h^2\bigl(\Gamma w^n+G^{n+1}\bigr),\\ \Gamma=\frac{\ell_g^2}{\nu\Delta t}.\end{gathered}
 \label{eq:step}
\end{equation}
The operator is symmetric positive definite, and its variable coefficient $h^2$ is supplied by the geometry. For $\Gamma=0$, Eq.~\eqref{eq:step} reduces to steady Poiseuille flow with a constant-coefficient operator, to which Corollary~\ref{cor:const} applies. For a coronary artery with $\ell_g\approx1.3$--$1.5$~mm, $\nu=3.3\times10^{-6}$~m$^2$/s, and one hundred steps per cardiac cycle ($\Delta t=8$~ms), we obtain $\Gamma\approx64$--$85$. We therefore study $\Gamma=0$, $100$, and $1000$.

\subsubsection{Elliptic duct}
Elliptic coordinates are $x=\ell_g\cosh\mu\cos\nu$, $y=\ell_g\sinh\mu\sin\nu$, with $h^2=\ell_g^2(\sinh^2\mu+\sin^2\nu)$, where $\ell_g$ is half the distance between the foci. The wall is $\mu=\mu_0$, with semi-axes $A_e=\ell_g\cosh\mu_0$ and $B_e=\ell_g\sinh\mu_0$. We use $B_e/A_e=0.5$. In order to avoid the branch line at the focal segment, we solve on the quarter domain $0\leq\mu\leq\mu_0$, $0\leq\nu\leq\pi/2$. The focal segment and the two axes are symmetry lines with zero normal derivative, and the wall carries $w=0$. On a cell-centered grid with $N=2^m$ points per direction, the one-dimensional Laplacians are diagonalized exactly by the DCT-IV in $\mu$ and the DCT-II in $\nu$. The measurement basis is $F=\mathrm{DCT\text{-}IV}\otimes\mathrm{DCT\text{-}II}$ on $2m$ qubits. The steady solution is known in closed form~\cite{White2006},
\begin{equation}
 w=\frac{A_e^2B_e^2}{2(A_e^2+B_e^2)}\Bigl(1-\frac{x^2}{A_e^2}-\frac{y^2}{B_e^2}\Bigr),\qquad Q=\frac{\pi A_e^3B_e^3}{4(A_e^2+B_e^2)}.
 \label{eq:ellipse}
\end{equation}
Since $h^2$ is a sum of a function of $\mu$ and a function of $\nu$, the operator of Eq.~\eqref{eq:step} is separable,
\begin{equation}
 \begin{gathered}A=A_\mu\otimes I+I\otimes A_\nu,\\ A_\mu=-\partial_{\mu\mu}+\Gamma\sinh^2\mu,\qquad A_\nu=-\partial_{\nu\nu}+\Gamma\sin^2\nu .\end{gathered}
 \label{eq:separable}
\end{equation}
The continuous counterparts of $A_\mu$ and $A_\nu$ are the modified and ordinary Mathieu operators~\cite{McLachlan1947}. Separability has three consequences. First, the system is solved classically by two $N\times N$ eigenvalue problems. Second, $e^{-iAT}=e^{-iA_\mu T}\otimes e^{-iA_\nu T}$, so no entangling gate between the registers is needed for the time evolution. Third, the source $h^2$ has Schmidt rank 2, and the product evolution preserves this rank. Each Krylov circuit therefore consists of one $R_Y$ rotation that prepares $s_1|0\rangle+s_2|1\rangle$ on the first qubit, one CNOT gate to the second register, and two $m$-qubit unitaries
\begin{equation}
 W_\mu=\mathrm{DCT_4}\,e^{-iA_\mu T}\,U_u,\qquad W_\nu=\mathrm{DCT_2}\,e^{-iA_\nu T}\,U_v,
 \label{eq:ellcirc}
\end{equation}
where $U_u$ and $U_v$ map $|0\rangle,|1\rangle$ to the Schmidt vectors of the source. The circuit is exact, and its depth does not depend on $T$. This makes the elliptic duct a structured calibration benchmark. On the other hand, generic synthesis of $W_\mu$ and $W_\nu$ has a cost exponential in the register size, so separability alone does not supply an efficient asymptotic quantum algorithm, and the problem is classically easy.

\subsubsection{Lens-shaped duct}
Bipolar coordinates are $x=\ell_g\sinh\tau/(\cosh\tau-\cos\sigma)$, $y=\ell_g\sin\sigma/(\cosh\tau-\cos\sigma)$, with $h=\ell_g/(\cosh\tau-\cos\sigma)$. A curve of constant $\sigma$ is a circular arc through the foci $(\pm\ell_g,0)$. The lens is the region $\sigma_0\leq\sigma\leq2\pi-\sigma_0$, whose corners at the foci have the interior angle $W_c=2\pi-2\sigma_0$. For $\sigma_0=\pi/2$ the lens is a disk of radius $\ell_g$. We use $\sigma_0=2\pi/3$, which gives corners of $120^\circ$. The strip $-\infty<\tau<\infty$ is truncated to $|\tau|\leq\tau_{\max}$ with $w=0$ at $\tau=\pm\tau_{\max}$. It can be shown that the homogeneous solutions decay as $e^{-\pi|\tau|/W_c}$ toward the corners, so the truncation error is exponentially small. We use $\tau_{\max}=6$. With Dirichlet conditions on all sides and $N$ interior points, the basis is $F=\mathrm{DST\text{-}I}\otimes\mathrm{DST\text{-}I}$.

The lens and the source are symmetric about $\sigma=\pi$ and $\tau=0$. Therefore only modes with odd indices in both directions are excited. This symmetry sector has dimension $N^2/4$. Symmetry alone does not make every coefficient in the sector nonzero. This is the selection rule of Proposition~\ref{prop:symmetry} for the two reflections of the lens. In the quantum register this sector corresponds to the least significant qubit of each register being 0, and a measured bitstring that violates this rule can be discarded. This parity filter plays the role of particle-number postselection in SQD~\cite{sqd}.

The lens operator is not separable. We use the split Trotter step
\begin{equation}
 \begin{gathered}e^{-iA\Delta t/r}\approx F_q^\dagger D_{\Lambda}F_qD_h,\\ D_{\Lambda}=e^{-i\Lambda_0\Delta t/r},\qquad D_h=e^{-i\Gamma\,\mathrm{diag}(h^2)\Delta t/r},\end{gathered}
 \label{eq:trotter}
\end{equation}
where $\Lambda_0$ contains the eigenvalues of the constant-coefficient Laplacian. Both factors are diagonal in some basis, so no Pauli decomposition of $A$ is needed, and $D_\Lambda$ is separable between the registers. The splitting is a Class II approximation. On hardware the grid and the transform must be chosen together. With $N=2^m$ interior points the DST-I has the denominator $N+1$, so a power-of-two QFT does not implement it exactly without padding. A cell-centered Dirichlet grid leads to a DST-II with a power-of-two denominator, but it changes the discrete operator and the source. We leave it as a future option. For small registers the DST-I can be synthesized exactly as a generic register unitary, and we use this in the hardware pilot. All numerical results below use the DST-I grid.

\subsubsection{General and curved cross sections}
For a patient-specific cross section, Step~1 is a numerical conformal map onto a rectangle or a disk~\cite{Driscoll2002,Marshall2007}. The analytic $h^2$ of the examples above is then replaced by a computed metric factor, and the map error is a Class I approximation bounded by Proposition~\ref{prop:perturb}. If no map is available, a Brinkman penalization on a bounding box~\cite{Angot1999} gives an operator of the same form with a discontinuous coefficient. For a curved vessel, Step~2 uses the coefficients~\eqref{eq:curved}. The curvature breaks the symmetry $\tau\to-\tau$ of the lens, so that the parity filter acts only in $\sigma$.

\section{Numerical results for duct flows}\label{sec:results}

\subsection{Setup}
All results in this section use exact time evolution of the discrete operator and multinomial sampling from the exact outcome laws, without hardware noise unless stated otherwise. We use $d=12$ Krylov states with equal shots per state. Unless stated otherwise, the sampling times are $T_k=k\pi/\lambda_{\max}$, the reference short step of Sec.~\ref{sec:support}. The error is the relative energy-norm error $\|x_{\mathcal S}-x^*\|_A/\|x^*\|_A$ of the Galerkin solution of Eq.~\eqref{eq:galerkin} with respect to the solution of the discrete system. We also use an oracle ordering, which selects modes by the magnitude of the exact coefficients $|\hat x_z|$. It is a reference ordering and need not minimize the Galerkin energy error. The essential set $\mathcal E_\varepsilon$ is the smallest set of the form required in Corollary~\ref{cor:complexity}, with $\varepsilon=1\%$, built in the oracle ordering. No hardware data are included.

\subsection{Fourier basis versus computational basis}
Figures~\ref{fig:lens} and~\ref{fig:ellipse} show the solutions and their Fourier coefficients. Table~\ref{tab:basis} compares the subspaces sampled on the $16\times16$ grids (8 qubits) with a total budget of $M=1200$ shots. In both geometries the Fourier basis gives an error that is one to three orders of magnitude smaller than the computational basis, with a subspace that is 1.7 to 7 times smaller. For $\Gamma=1000$ the time evolution helps clearly. With twelve states the error is $3.3\times10^{-4}$ (lens) and $8.3\times10^{-3}$ (ellipse), while the single state $|b\rangle$ with the same number of shots gives $2.4\times10^{-2}$ and $2.3\times10^{-2}$.

\begin{figure*}[t]
\centering
\includegraphics[width=\textwidth]{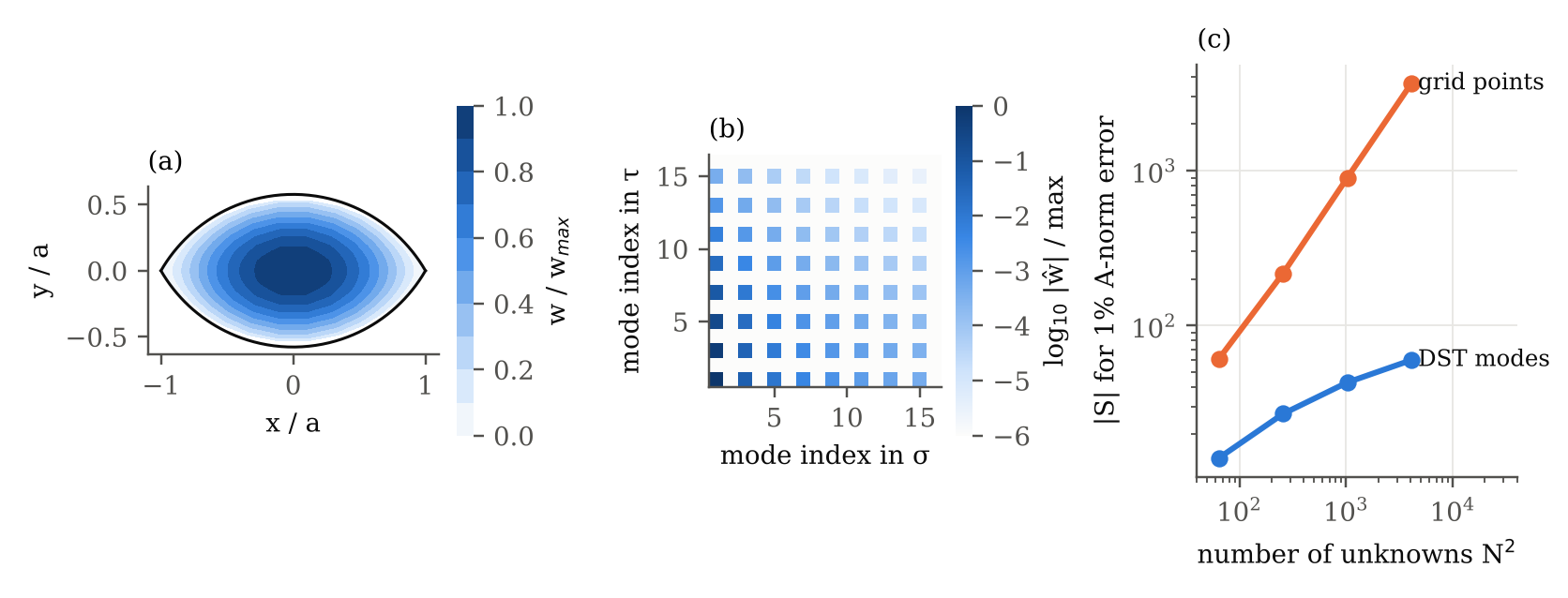}
\caption{Lens-shaped duct ($\sigma_0=2\pi/3$, $32\times32$ grid). (a) Steady axial velocity. (b) Magnitude of the DST-I coefficients of the solution; only odd--odd modes are nonzero. (c) Number of basis states required for $1\%$ error ($\Gamma=100$) versus the number of unknowns: DST modes selected by $|\hat b|$, grid points selected by an oracle ordering, and the total number of unknowns (dashed).}\label{fig:lens}
\end{figure*}

\begin{figure*}[t]
\centering
\includegraphics[width=\textwidth]{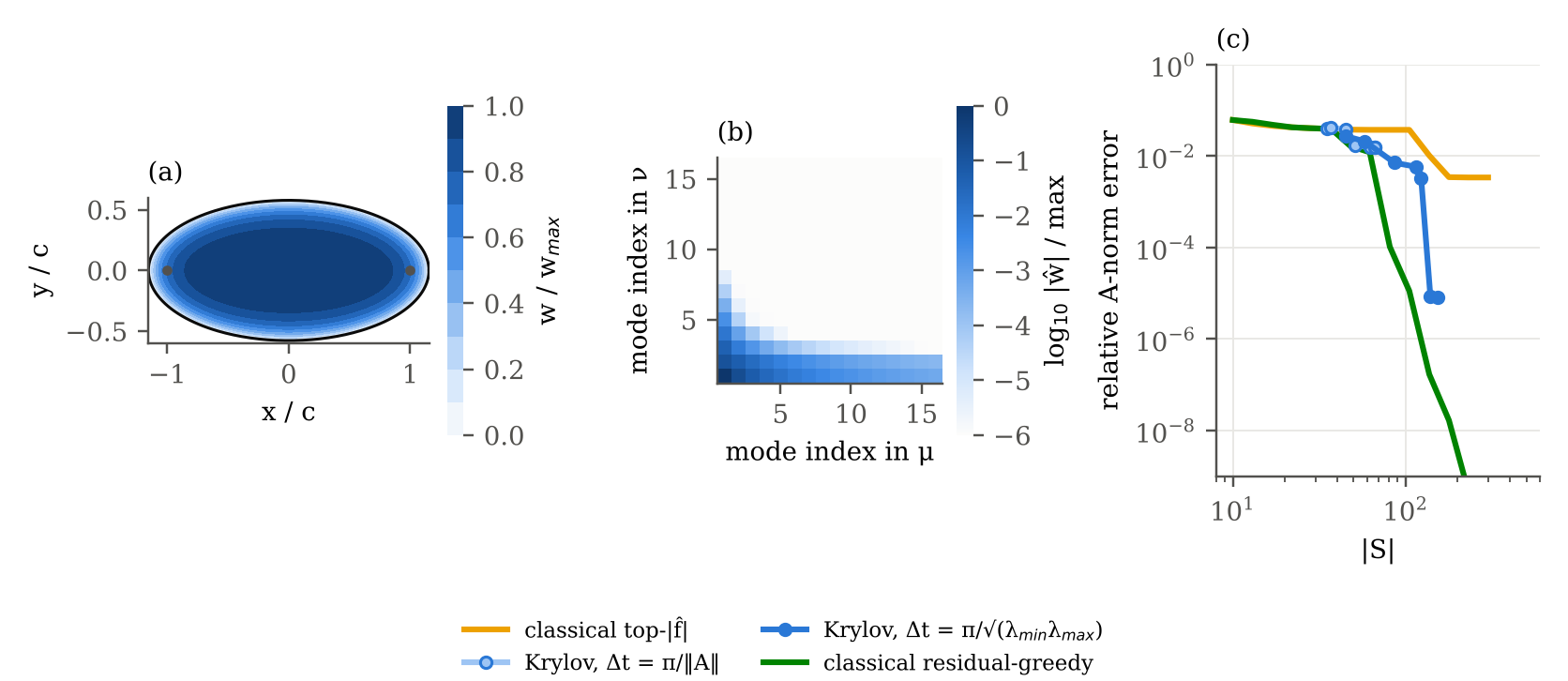}
\caption{Elliptic duct ($B_e/A_e=0.5$, $32\times32$ quarter grid, $\Gamma=100$). (a) Axial velocity after the first implicit step from rest, mirrored to the full ellipse; dots mark the foci. (b) Magnitude of the coefficients in the DCT-IV$\otimes$DCT-II basis; the solution is concentrated near two lines of modes. (c) Error versus subspace size for classical selection by $|\hat b|$, Krylov sampling with two uniform steps ($\pi/\lambda_{\max}$ and $\pi/\sqrt{\lambda_{\min}\lambda_{\max}}$; up to 19200 shots), and classical residual-greedy selection.}\label{fig:ellipse}
\end{figure*}

\begin{table*}[t]
\centering
\caption{Sampled subspaces on the $16\times16$ grids (8 qubits), $M=1200$, $T_k=k\pi/\lambda_{\max}$. Entries are $|\mathcal S|$ / relative energy-norm error.}\label{tab:basis}
\begin{tabular}{llccc}
\toprule
Geometry & $\Gamma$ & Fourier, $d=12$ & Fourier, $d=1$ & Computational, $d=12$\\
\midrule
Lens & 0 & 26 / $1.0\times10^{-2}$ & 24 / $1.8\times10^{-2}$ & 139 / $1.8\times10^{-1}$\\
Lens & 100 & 41 / $1.0\times10^{-2}$ & 21 / $2.3\times10^{-2}$ & 136 / $2.5\times10^{-1}$\\
Lens & 1000 & 54 / $3.3\times10^{-4}$ & 23 / $2.4\times10^{-2}$ & 93 / $3.0\times10^{-1}$\\
Ellipse & 0 & 28 / $4.4\times10^{-3}$ & 25 / $6.4\times10^{-3}$ & 192 / $5.3\times10^{-1}$\\
Ellipse & 100 & 28 / $2.1\times10^{-2}$ & 25 / $4.1\times10^{-2}$ & 195 / $4.9\times10^{-1}$\\
Ellipse & 1000 & 60 / $8.3\times10^{-3}$ & 28 / $2.3\times10^{-2}$ & 193 / $3.6\times10^{-1}$\\
\bottomrule
\end{tabular}
\end{table*}

Table~\ref{tab:scaling} shows the scaling with the grid for $\Gamma=100$. For the lens, the number of modes selected by $|\hat b|$ for $1\%$ error grows from 14 to 60 while the number of unknowns grows from 64 to 4096. For the ellipse the oracle ordering requires between 19 and 45 modes. In the computational basis, an oracle ordering still requires $85$--$95\%$ of all grid points for the lens and all grid points for the ellipse. The selection by $|\hat b|$ is smaller than the oracle ordering in some lens rows, which illustrates that the oracle ordering is not an optimal Galerkin subspace.

\begin{table*}[t]
\centering
\caption{Scaling with the grid ($\Gamma=100$). Number of basis states required for $1\%$ relative energy-norm error. The lens grid is the full domain; the ellipse grid is the quarter domain.}\label{tab:scaling}
\begin{tabular}{rrccccc}
\toprule
& & \multicolumn{3}{c}{Lens} & \multicolumn{2}{c}{Ellipse}\\
\cmidrule(lr){3-5}\cmidrule(lr){6-7}
$N$ & qubits & oracle & top-$|\hat b|$ & grid points (oracle) & oracle & top-$|\hat b|$\\
\midrule
8 & 6 & 14 & 14 & 61 & 19 & 27\\
16 & 8 & 27 & 27 & 217 & 32 & 115\\
32 & 10 & 46 & 43 & 897 & 45 & 143\\
64 & 12 & 63 & 60 & 3649 & 40 & 249\\
\bottomrule
\end{tabular}
\end{table*}

\subsection{Verification of the energy-norm theory}\label{sec:theorycheck}
Supplementary Sec.~\ref{sup:verification} tests the theory of Sec.~\ref{sec:galerkin-theory} directly. The support-completeness bound~\eqref{eq:support} holds for every mode. The single-state law, by contrast, vanishes on solution modes of the ellipse that its source does not excite, which explains the failure of the selection by $|\hat b|$ for the ellipse in Table~\ref{tab:scaling}. The bound of Theorem~\ref{thm:error} holds in all cases and exceeds the observed root-mean-square error by a factor of 1.2 to 18. The short step $\pi/\lambda_{\max}$ fails on the finer elliptic grids, where $q_{\min}(\mathcal E_\varepsilon)$ drops to $10^{-15}$--$10^{-20}$, while the multiscale heuristic is within a factor of 2 of the better uniform step in all twelve cases tested. With this heuristic the efficiency $\eta_{\mathrm s}$ lies between $2\times10^{-3}$ and $6\times10^{-2}$. Initial states that resemble $|\hat x|^2$ were harmful, because efficient coverage requires a flat law on $\mathcal E_\varepsilon$ (Proposition~\ref{prop:eff}).

\subsection{Noise and the parity filter}
We model hardware errors by replacing each shot, with probability $p$, by a uniformly random index. Table~\ref{tab:noise} shows the result for the lens on the $32\times32$ grid. The parity filter removes about three quarters of the spurious indices without changing the error. In these runs the error decreases as the contamination increases, because random indices fill the small symmetric sector. This is an empirical outcome for this law. It is not a consequence that holds at a fixed shot budget for arbitrary replacement noise (Corollary~\ref{cor:noise}), and it shows that a low error alone is not evidence of a useful quantum sampler.

\begin{table}[t]
\centering
\caption{Effect of sample contamination (lens, $32\times32$, $\Gamma=100$, $M=4800$, mean of five runs). Entries are $|\mathcal S|$ / error.}\label{tab:noise}
\begin{tabular}{ccc}
\toprule
$p$ & no filter & parity filter\\
\midrule
0 & 58 / $8.0\times10^{-3}$ & 58 / $8.0\times10^{-3}$\\
0.05 & 250 / $6.5\times10^{-3}$ & 95 / $6.5\times10^{-3}$\\
0.2 & 650 / $5.2\times10^{-3}$ & 177 / $5.2\times10^{-3}$\\
0.5 & 927 / $2.3\times10^{-3}$ & 236 / $2.3\times10^{-3}$\\
\bottomrule
\end{tabular}
\end{table}

\subsection{Comparison with classical mode selection}
Because $b$ and $A$ are known classically, the sampler must be compared with classical selections of the same Galerkin subspace. We consider three. The first selects the modes with the largest $|\hat b_z|$. The second selects the largest $|\hat b_z/\Lambda_{0,z}|$. The third is the residual-greedy method, which starts from the ten largest $|\hat b_z|$ and repeatedly adds the modes with the largest residual, i.e., Phase~2 of Algorithm~\ref{alg:solver} used alone. Its cost includes residual formation and repeated solves; the efficiency $\eta_{\mathrm s}$ is not a performance measure for this deterministic method.

Table~\ref{tab:classical} shows the result for the ellipse at matched $|\mathcal S|$. With a suitable uniform step, the Krylov sampling outperforms the first two selections by one to four orders of magnitude. This is the effect of Theorem~\ref{thm:support}. The residual-greedy method, however, is better than the Krylov sampling in every case, often by several orders of magnitude. For the lens it is at least as good. For an off-center localized source in the lens, the selection by $|\hat b|$ was also better than the Krylov sampling, by a factor of 1.3 to 2.4 in error. These results are consistent with the efficiency analysis and with the remark in Sec.~\ref{sec:discussion} that favorable concentration alone is not evidence of quantum advantage.

\begin{table*}[t]
\centering
\caption{Ellipse, $\Gamma=100$, $M=19200$: Krylov sampling with $T_k=k\,c\,\pi/\lambda_{\max}$ versus classical selections at the same $|\mathcal S|$. Entries are relative energy-norm errors.}\label{tab:classical}
\begin{tabular}{rrcccc}
\toprule
$N$ & $c$ & Krylov ($|\mathcal S|$ / error) & top-$|\hat b|$ & top-$|\hat b/\Lambda_0|$ & residual-greedy\\
\midrule
16 & 1 & 37 / $1.2\times10^{-2}$ & $1.6\times10^{-2}$ & $1.0\times10^{-2}$ & $1.9\times10^{-3}$\\
16 & 10 & 93 / $1.1\times10^{-6}$ & $1.0\times10^{-2}$ & $9.7\times10^{-3}$ & $6.0\times10^{-8}$\\
32 & 1 & 64 / $1.5\times10^{-2}$ & $3.8\times10^{-2}$ & $3.8\times10^{-2}$ & $1.2\times10^{-2}$\\
32 & 10 & 90 / $1.1\times10^{-3}$ & $3.8\times10^{-2}$ & $3.5\times10^{-3}$ & $1.7\times10^{-5}$\\
32 & 100 & 192 / $9.8\times10^{-4}$ & $3.5\times10^{-3}$ & $1.6\times10^{-3}$ & $7.9\times10^{-9}$\\
64 & 1 & 110 / $3.8\times10^{-2}$ & $3.8\times10^{-2}$ & $3.8\times10^{-2}$ & $4.6\times10^{-3}$\\
64 & 100 & 218 / $1.6\times10^{-3}$ & $3.7\times10^{-2}$ & $5.3\times10^{-4}$ & $6.3\times10^{-7}$\\
\bottomrule
\end{tabular}
\end{table*}

\subsection{Geometric extensions}\label{sec:geomresults}
\paragraph{Curvature} Table~\ref{tab:curv} applies Eq.~\eqref{eq:curved} to the disk and the $120^\circ$ lens on the $32\times32$ grid with $\Gamma=100$, the multiscale heuristic, and $M=19200$. The curvature breaks the symmetry $\tau\to-\tau$, and the number of nonzero modes doubles from 256 to 512. The essential set grows by a factor of 1.5 to 1.9 at $\kappa_c\ell_g=0.3$, while $\eta_{\mathrm s}$ stays of order $10^{-3}$. The source $h^2$ remains symmetric while the solution does not, which is the support mismatch of Theorem~\ref{thm:support}. As a result, at $\kappa_c\ell_g=0.3$ the Krylov sampling is 2.5 to 5 times more accurate than the selection by $|\hat b|$. The residual-greedy selection remains more accurate than both.

\begin{table*}[t]
\centering
\small\setlength{\tabcolsep}{4pt}
\caption{Curved vessels, reduced model~\eqref{eq:curved} ($32\times32$, $\Gamma=100$, multiscale heuristic, $M=19200$). Errors are relative energy-norm errors at the $|\mathcal S|$ reached by the Krylov sampling.}\label{tab:curv}
\begin{tabular}{lrrrcrccc}
\toprule
Cross section & $\kappa_c\ell_g$ & nonzero & $|\mathcal E_\varepsilon|$ & $\eta_{\mathrm s}$ & $|\mathcal S|$ & Krylov & top-$|\hat b|$ & residual-greedy\\
\midrule
Disk & 0 & 256 & 70 & $6.0\times10^{-3}$ & 236 & $2.0\times10^{-3}$ & $6.4\times10^{-7}$ & $3.5\times10^{-7}$\\
Disk & 0.1 & 512 & 107 & $5.6\times10^{-3}$ & 401 & $4.3\times10^{-3}$ & $3.4\times10^{-3}$ & $2.0\times10^{-6}$\\
Disk & 0.3 & 512 & 132 & $5.8\times10^{-3}$ & 413 & $4.0\times10^{-3}$ & $1.0\times10^{-2}$ & $3.3\times10^{-6}$\\
Lens $120^\circ$ & 0 & 256 & 58 & $6.7\times10^{-3}$ & 190 & $1.5\times10^{-3}$ & $2.5\times10^{-6}$ & $9.3\times10^{-7}$\\
Lens $120^\circ$ & 0.1 & 512 & 84 & $4.4\times10^{-3}$ & 338 & $2.9\times10^{-3}$ & $7.6\times10^{-3}$ & $2.5\times10^{-6}$\\
Lens $120^\circ$ & 0.3 & 511 & 102 & $2.0\times10^{-3}$ & 338 & $4.5\times10^{-3}$ & $2.3\times10^{-2}$ & $4.2\times10^{-6}$\\
\bottomrule
\end{tabular}
\end{table*}

\subsection{Larger and multiscale examples}\label{sec:large}
The examples above use at most $32\times32$ modes. With Krylov sampling at this size, the sampled set covers a large fraction of the admissible space. We therefore simulated the ideal and the noisy samplers on grids of up to $256\times256$ modes (16 data qubits, $n=65536$). We considered two non-separable problems. The first is the curved $120^\circ$ lens of Table~\ref{tab:curv} with $\kappa_c\ell_g=0.3$, $\Gamma=100$, and the source $h^2$. The second is a $90^\circ$ bend of a channel, $1\le r\le2$, $0\le\theta\le\pi/2$. In the log-polar coordinates $(\xi,\theta)=(\ln r,\theta)$ the map is conformal with $h^2=r^2$, and the operator takes the form~\eqref{eq:curved} with $\Gamma=100$ on a rectangle. Its source is $h^2\omega$, where $\omega(x,y)$ is a synthetic two-dimensional turbulent vorticity field: a superposition of 400 plane waves with random directions and phases, wavenumbers distributed log-uniformly in $[2,48]$, and an energy spectrum proportional to $k^{-3}$. The source therefore has energy on all scales of the grid.

For each problem we used the multiscale schedule with $d=12$ Krylov states, $M$ between $3\times10^2$ and $10^5$ samples (three repetitions each), and replacement noise $p=0.3$ for the noisy sampler. The time evolution was computed exactly with a Krylov exponential integrator, and the Galerkin problem on the sampled set was solved without forming a matrix, by conjugate gradients in the DST basis. The classical selections were evaluated at the same subspace sizes. Table~\ref{tab:large} and Fig.~\ref{fig:large} summarize the results.

\begin{table*}[t]
\centering
\caption{Larger examples. For each grid, the Krylov sampler used $M=10^5$ samples; $|\mathcal S|$ is the mean size of the sampled set and its fraction of $n$. The other selections are evaluated at the same $|\mathcal S|$, except the noisy sampler, which used the same $M$ and gives a larger set (fraction in parentheses). Errors are relative energy-norm errors.}\label{tab:large}
\resizebox{\linewidth}{!}{\begin{tabular}{llrrcrccccc}
\toprule
Problem & Grid (qubits) & $|\mathcal E_\varepsilon|$ & $\kappa$ & $\eta_{\mathrm s}$ & $|\mathcal S|$ ($|\mathcal S|/n$) & Krylov & Krylov, $p=0.3$ & top-$|\hat b|$ & residual-greedy & oracle\\
\midrule
Curved lens & 32$^2$ (10) & 102 & 605 & $2.0\times10^{-3}$ & 437 (42.6\%) & $1.2\times10^{-3}$ & $2.8\times10^{-12}$ (100\%) & $1.0\times10^{-2}$ & $2.3\times10^{-7}$ & $1.0\times10^{-7}$\\
 & 64$^2$ (12) & 157 & 2351 & $3.8\times10^{-3}$ & 1119 (27.3\%) & $1.9\times10^{-3}$ & $1.9\times10^{-11}$ (100\%) & $5.5\times10^{-3}$ & $1.8\times10^{-7}$ & $1.5\times10^{-7}$\\
 & 128$^2$ (14) & 228 & 9269 & $3.2\times10^{-3}$ & 1774 (10.8\%) & $2.1\times10^{-3}$ & $9.2\times10^{-4}$ (85\%) & $1.3\times10^{-1}$ & $8.1\times10^{-7}$ & $9.7\times10^{-7}$\\
 & 256$^2$ (16) & 311 & 36810 & $2.2\times10^{-3}$ & 2079 (3.2\%) & $2.1\times10^{-3}$ & $2.0\times10^{-3}$ (39\%) & $1.3\times10^{-1}$ & $2.0\times10^{-5}$ & $1.3\times10^{-5}$\\
\midrule
Curved channel & 32$^2$ (10) & 669 & 55 & $6.3\times10^{-4}$ & 697 (68.1\%) & $7.3\times10^{-3}$ & $7.0\times10^{-12}$ (100\%) & $9.2\times10^{-3}$ & $5.6\times10^{-3}$ & $4.8\times10^{-3}$\\
 & 64$^2$ (12) & 1477 & 210 & $2.9\times10^{-4}$ & 1192 (29.1\%) & $1.2\times10^{-2}$ & $2.6\times10^{-7}$ (100\%) & $1.2\times10^{-2}$ & $8.8\times10^{-3}$ & $7.9\times10^{-3}$\\
 & 128$^2$ (14) & 2826 & 823 & $9.0\times10^{-5}$ & 1642 (10.0\%) & $1.3\times10^{-2}$ & $5.9\times10^{-3}$ (86\%) & $1.1\times10^{-2}$ & $9.6\times10^{-3}$ & $8.3\times10^{-3}$\\
 & 256$^2$ (16) & 5222 & 3262 & $8.2\times10^{-5}$ & 2095 (3.2\%) & $1.3\times10^{-2}$ & $1.3\times10^{-2}$ (39\%) & $1.1\times10^{-2}$ & $9.5\times10^{-3}$ & $7.8\times10^{-3}$\\
\bottomrule
\end{tabular}}
\end{table*}

\begin{figure}[t]
\centering
\includegraphics[width=\linewidth]{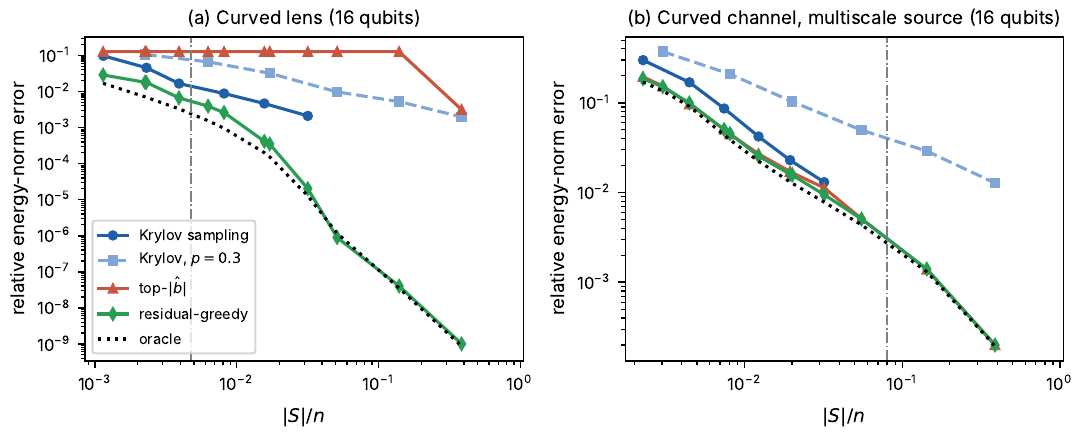}
\caption{Relative energy-norm error against the fraction $|\mathcal S|/n$ of modes retained on the $256\times256$ grids (16 qubits, $n=65536$). The dash-dotted line marks $|\mathcal E_\varepsilon|/n$. (a) Curved lens. (b) Curved channel with the multiscale source.}\label{fig:large}
\end{figure}

\paragraph{Curved lens} The admissible dimension grows by a factor of 64 from $32\times32$ to $256\times256$, while the essential set grows only by a factor of about three. As a result, the fraction of the space sampled with $M=10^5$ decreases from 43\% to 3.2\% at a nearly constant error of about $2\times10^{-3}$. On the $256\times256$ grid, $10^4$ samples give an error of $8.8\times10^{-3}$ with 0.8\% of the modes. The sampling efficiency $\eta_{\mathrm s}$ stays between $2\times10^{-3}$ and $4\times10^{-3}$, which is the behavior expected from Corollary~\ref{cor:complexity} when the outcome law concentrates on a slowly growing set. The support mismatch also persists. The selection by $|\hat b|$ stalls at an error of $0.13$ on the $256\times256$ grid while it retains up to 14\% of the modes, whereas the Krylov sampling reaches $2\times10^{-3}$ with 3.2\%. The residual-greedy selection, however, is still two to four orders of magnitude more accurate than the Krylov sampling at the same size, because the solution remains smooth and its essential modes are found by a few residual evaluations.

\paragraph{Curved channel with a multiscale source} Here the source excites all scales, and the essential set is large: 17\% of the admissible space on the $128\times128$ grid and 8\% on the $256\times256$ grid. On the largest grid the Krylov sampling reaches an error of $1.3\times10^{-2}$ with 3.2\% of the modes. At every subspace size, the error of the Krylov sampling exceeds that of the oracle selection by a factor of 1.6 to 2.0 and that of the residual-greedy selection by a factor of 1.4 to 1.7. The advantage of the classical selections is therefore much smaller than for the ducts. The selection by $|\hat b|$ performs as well as the residual-greedy selection, because the source and the solution have similar mode structures; condition~(1) of Sec.~\ref{sec:conclusion} is therefore not satisfied. The sampling efficiency decreases from $6\times10^{-4}$ to $8\times10^{-5}$ with the grid size, and the noisy sampler fills most of the space before it reaches the accuracy of the ideal sampler.

These results support the scaling of the analysis with the register size and show that the sampled set becomes a small fraction of the space. They also show that, for the elliptic problems considered here, a classical selection that uses the residual remains competitive or better. A benefit requires problems in which the essential set is both hard to identify classically and different from the support of the source.

\section{Hardware validation on two platforms}\label{sec:prospective}\label{sec:hardware}
We tested the method on two processors of different technology: the IonQ Forte-1 trapped-ion processor and the IBM superconducting processor \texttt{ibm\_cleveland}. We used the $8\times8$ grids (six data qubits) of the ellipse and the lens with $\Gamma=100$ on both processors, and the $16\times16$ grids (eight data qubits) on the IBM processor. The six-qubit logical circuits are identical on the two platforms. At these sizes every quantity of Sec.~\ref{sec:results} can be computed exactly. The measured outcome laws can therefore be compared with the ideal laws mode by mode. The experiment does not test asymptotic scaling, and it is not intended to show an advantage. It tests four elements that a larger implementation will need: the measurement convention, the noise model of the outcome laws, the geometric parity rule as an error-detection mechanism, and the Galerkin reconstruction from hardware samples. With 1500 shots per case, which exceeds the dimension, the sampled sets cover 89 to 100\% of the admissible basis. The experiments therefore test the outcome laws, the noise model, and the selection rule. They do not test the reduction of the subspace dimension, which is the purpose of sample-based methods.

\subsection{Circuits and execution}
We used three Krylov states for each duct, with $T_k=k\Delta t$, $k=0,1,2$, and the geometric-mean step $\Delta t=\pi/\sqrt{\lambda_{\min}\lambda_{\max}}$. For the ellipse the operator separates as $A=A_\mu\otimes I+I\otimes A_\nu$, and the source $h^2$ has Schmidt rank two. Each circuit then consists of a short preparation of the two Schmidt coefficients and one unitary per register. The register unitary combines the Schmidt vectors of $b$, the evolution, and the transform to the cosine basis. The circuits are exact. They contain 38 or 39 logical two-qubit gates for six qubits, below the bound of 41 from generic three-qubit synthesis~\cite{Shende2006}, and 191 for eight qubits. These register unitaries are computed from the classical diagonalization of one-dimensional operators. The ellipse circuits therefore test sampling, measurement, and reconstruction under real noise, not a scalable implementation of the time evolution.

For the lens we prepared the source exactly and used one symmetric Trotter step of the split evolution of Sec.~\ref{sec:ducts} per circuit. A step consists of register unitaries built from the DST-I and $e^{-i\Lambda\delta/2}$, and the diagonal phase $e^{-i\Gamma h^2\delta}$ on all qubits. These circuits use no eigendecomposition of $A$, and all qubits are active. They contain 78, 194, and 315 logical two-qubit gates for six qubits and 332, 862, and 1395 for eight qubits. With a single step, the circuit laws for $k=1$ and 2 differ from the exact evolution by total-variation distances between 0.05 and 0.08. We therefore compare the measured laws with the laws of the circuits as submitted.

We use the register convention of Supplementary Sec.~\ref{app:discretization}: the integer outcome of the measured bits equals the row-major mode index $k_1N+k_2$. Every circuit was checked against its ideal law by an exact statevector calculation before submission. Each circuit was run with 500 shots. The IonQ circuits were run on 11 October 2026 without debiasing or symmetry verification. The provider returned the outcome frequencies; all of them are integer multiples of $1/500$, so the integer counts are recovered exactly. The IBM circuits were run on 14 October 2026 with dynamical decoupling and twirling disabled, and the raw counts were recorded. The IBM processor uses CZ as its native two-qubit gate. On the qubits used, the calibrated median CZ error was $1.5\times10^{-3}$ and the median readout error was $3.5\times10^{-3}$. Routing to the heavy-hex connectivity increased the native two-qubit count by a factor of 1.6 to 2.1 relative to the logical count.

\begin{table*}[t]
\centering
\caption{Hardware results (500 shots per circuit). 2Q: logical / native two-qubit gates. TV: total-variation distance between the measured law and the circuit law; the sampling column is its expected value for 500 ideal shots. $p$: fitted replacement probability of the model $(1-p)q+p/2^n$. Residual: distance after the fit, with its sampling value in parentheses. Acceptance: fraction of lens shots in the odd--odd sector, followed by the value $1-3p/4$ predicted by the fitted model.}\label{tab:hw-results}
\footnotesize\setlength{\tabcolsep}{3.5pt}
\resizebox{\linewidth}{!}{\begin{tabular}{llrccccc}
\toprule
Platform & Circuit & 2Q & TV & TV (sampling) & $p$ & Residual & Acceptance / $1-3p/4$\\
\midrule
IonQ, 6 qubits & Ellipse, $k=0$ & 38 & 0.296 & 0.046 & 0.38 & 0.215 (0.110) & --\\
 & Ellipse, $k=1$ & 39 & 0.274 & 0.062 & 0.28 & 0.191 (0.110) & --\\
 & Ellipse, $k=2$ & 39 & 0.226 & 0.074 & 0.23 & 0.139 (0.109) & --\\
 & Lens, $k=0$ & 78 & 0.327 & 0.043 & 0.42 & 0.233 (0.114) & 0.726 / 0.689\\
 & Lens, $k=1$ & 194 & 0.539 & 0.055 & 0.64 & 0.167 (0.130) & 0.510 / 0.520\\
 & Lens, $k=2$ & 315 & 0.748 & 0.042 & 0.88 & 0.209 (0.139) & 0.310 / 0.338\\
\midrule
IBM, 6 qubits & Ellipse, $k=0$ & 38 / 67 & 0.441 & 0.046 & 0.51 & 0.296 (0.120) & --\\
 & Ellipse, $k=1$ & 39 / 68 & 0.268 & 0.062 & 0.29 & 0.178 (0.111) & --\\
 & Ellipse, $k=2$ & 39 / 68 & 0.258 & 0.073 & 0.14 & 0.214 (0.099) & --\\
 & Lens, $k=0$ & 78 / 143 & 0.686 & 0.043 & 0.76 & 0.273 (0.134) & 0.462 / 0.431\\
 & Lens, $k=1$ & 194 / 317 & 0.779 & 0.055 & 0.98 & 0.186 (0.142) & 0.262 / 0.266\\
 & Lens, $k=2$ & 315 / 527 & 0.856 & 0.043 & 0.97 & 0.149 (0.143) & 0.266 / 0.270\\
\midrule
IBM, 8 qubits & Ellipse, $k=0$ & 191 / 340 & 0.677 & 0.055 & 0.78 & 0.447 (0.272) & --\\
 & Ellipse, $k=1$ & 191 / 340 & 0.660 & 0.064 & 0.78 & 0.425 (0.273) & --\\
 & Ellipse, $k=2$ & 191 / 340 & 0.663 & 0.074 & 0.71 & 0.436 (0.266) & --\\
 & Lens, $k=0$ & 332 / 711 & 0.911 & 0.057 & 1.00 & 0.298 (0.276) & 0.296 / 0.250\\
 & Lens, $k=1$ & 862 / 1813 & 0.900 & 0.088 & 1.00 & 0.276 (0.276) & 0.216 / 0.251\\
 & Lens, $k=2$ & 1395 / 2975 & 0.890 & 0.090 & 1.00 & 0.295 (0.277) & 0.260 / 0.251\\
\bottomrule
\end{tabular}}
\end{table*}

\subsection{Analysis}
At these sizes coverage of the essential set is not a demanding test. With the three six-qubit circuits, the ideal sampler finds all essential modes with probability 0.94 (ellipse, $|\mathcal E_\varepsilon|=21$) and 0.84 (lens, $|\mathcal E_\varepsilon|=15$) from 1500 shots, and a uniform sampler needs only about 350 and 78 shots, respectively, for probability 0.9. We therefore base the analysis on the outcome laws and on the noise model of Corollary~\ref{cor:noise}. We then compare the reconstructions at equal subspace size.

\paragraph{Outcome laws} Table~\ref{tab:hw-results} and Fig.~\ref{fig:ibm} and Supplementary Fig.~\ref{fig:ionq} show that the measured laws differ from the circuit laws by much more than the sampling fluctuation on both platforms. The survival fraction $1-p$ of the fitted replacement model decreases with circuit size. With $n_{2\mathrm q}$ the number of logical two-qubit gates, a fit of the form $1-p=c\,e^{-\alpha n_{2\mathrm q}}$ [Fig.~\ref{fig:ibm}(b)] gives $c=0.93$ and $\alpha=0.0061$ for IonQ, and $c=0.72$ and $\alpha=0.0060$ for IBM, where the IBM fit excludes the saturated circuits with $p>0.95$. The decay per logical gate is thus the same within the accuracy of the fits. Per native CZ gate the IBM rate is $0.0034$, about twice the calibrated gate error, which indicates additional idling and crosstalk errors in the absence of dynamical decoupling. The two platforms differ mainly in the gate-independent loss, about $7.5\%$ for IonQ and about $28\%$ for IBM. The residual after the fit is 1.0 to 2.5 times its sampling value. Thus the noise is not entirely uniform, and the deviation is largest for the strongly concentrated law of the ellipse at $k=0$. For IonQ, a model with independent bit flips on each qubit in addition to uniform replacement reduces the residual to 1.1 to 1.5 times its sampling value, but the fitted flip probabilities vary from circuit to circuit and are not determined reliably by 500 shots. The eight-qubit lens circuits, with 711 to 2975 native two-qubit gates, are fully randomized ($p>0.99$).

\paragraph{Parity filter} For the lens, the measured fraction of shots in the odd--odd sector agrees with the prediction $1-3p/4$ of the fitted uniform model to within 0.04 for all six-qubit circuits on both platforms. For the fully randomized eight-qubit circuits it approaches the uniform value $1/4$, as expected. On IonQ the filter removes about half of all shots and reduces the distance to the circuit law from 0.33 to 0.17 for $k=0$ and from 0.54 to 0.16 for $k=1$. The parity filter therefore acts as a direct and inexpensive noise diagnostic, as anticipated in Sec.~\ref{sec:symmetry}. This use of a measurement-visible constraint is the same mechanism as the particle-number recovery of SQD~\cite{sqd} and the code-space recovery of Ref.~\cite{csr}. In the present case the constraint is not engineered; it follows from the isometries of the cross section.

\paragraph{Dominance} Corollary~\ref{cor:noise} predicts $q_{\mathrm{hw}}(z)\geq(1-p)q(z)$ for replacement noise. For the pooled six-qubit lens shots, $\min_{z\in\mathcal E_\varepsilon}q_{\mathrm{hw}}(z)/q(z)$ is 0.36 on IonQ and 0.09 on IBM, compared with $1-\bar p=0.35$ and 0.10. For the ellipse the ratio falls below $1-\bar p$ on both platforms: 0.65 versus 0.70 on IonQ, 0.42 versus 0.69 on IBM with six qubits, and 0.14 versus 0.25 on IBM with eight qubits. These violations are consistent with the nonuniform part of the noise, for which the pointwise dominance is not guaranteed. The coverage bounds of Sec.~\ref{sec:galerkin-theory} should therefore be applied to hardware data with the measured, not the modeled, outcome law.

\begin{figure*}[t]
\centering
\includegraphics[width=\textwidth]{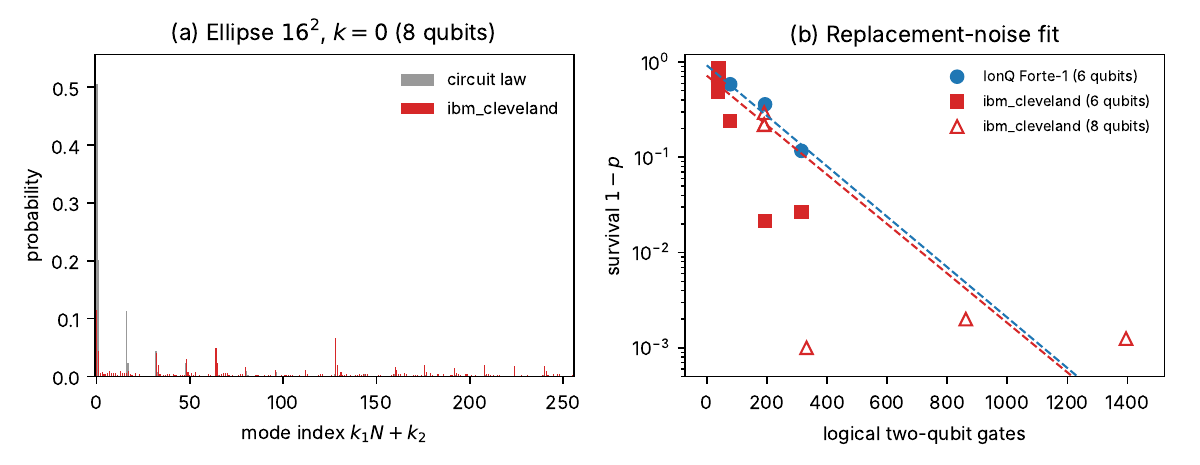}
\caption{IBM \texttt{ibm\_cleveland} results. (a) Circuit law (gray) and measured frequencies (red, 500 shots) for the eight-qubit ellipse at $k=0$. (b) Survival fraction $1-p$ of the fitted replacement model versus the number of logical two-qubit gates on both platforms, with exponential fits (dashed). Points at the bottom are fully randomized circuits ($p>0.99$).}\label{fig:ibm}
\end{figure*}

\begin{table*}[t]
\centering
\caption{Subspace size and Galerkin error for the full hardware samples (1500 shots per case). Admissible: dimension of the symmetry sector. Used: observed admissible modes. The last three columns give the relative energy error of the hardware set and of classical selections of the same size. Values near $10^{-15}$ are at the level of rounding.}\label{tab:equal-size}
\footnotesize
\begin{tabular}{llrrcccc}
\toprule
Platform & Case & Admissible & Used & Fraction & Hardware & Residual-greedy & Oracle\\
\midrule
IonQ & Ellipse $8^2$ & 64 & 60 & 0.94 & $1.5\times10^{-6}$ & $1.3\times10^{-8}$ & $1.3\times10^{-8}$\\
IonQ & Lens $8^2$ & 16 & 16 & 1.00 & $4\times10^{-16}$ & $4\times10^{-16}$ & $4\times10^{-16}$\\
IBM & Ellipse $8^2$ & 64 & 61 & 0.95 & $2.0\times10^{-5}$ & $6.9\times10^{-9}$ & $6.9\times10^{-9}$\\
IBM & Lens $8^2$ & 16 & 16 & 1.00 & $4\times10^{-16}$ & $4\times10^{-16}$ & $4\times10^{-16}$\\
IBM & Ellipse $16^2$ & 256 & 227 & 0.89 & $7.1\times10^{-7}$ & $2\times10^{-15}$ & $2\times10^{-15}$\\
IBM & Lens $16^2$ & 64 & 63 & 0.98 & $2.3\times10^{-3}$ & $9.8\times10^{-6}$ & $9.8\times10^{-6}$\\
\bottomrule
\end{tabular}
\end{table*}

\begin{figure*}[t]
\centering
\includegraphics[width=\textwidth]{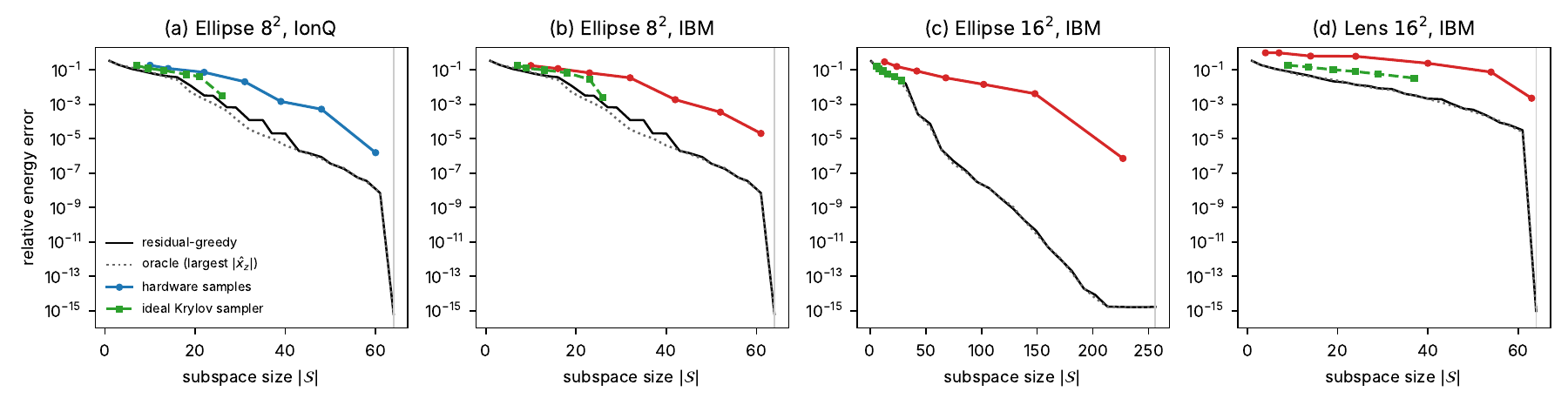}
\caption{Relative energy error of the Galerkin solution versus subspace size. Hardware samples (circles) are obtained by subsampling the measured counts at budgets of 15 to 480 shots, with the full sample as the last point; the parity filter is applied for the lens. Squares: ideal Krylov sampler with the same circuits and budgets. Solid line: classical residual-greedy selection. Dotted line: oracle selection of the largest $|\hat x_z|$. The vertical line marks the dimension of the admissible space.}\label{fig:equal-size}
\end{figure*}

\paragraph{Subspace size and reconstruction} Table~\ref{tab:equal-size} shows that the full hardware samples contain 89 to 100\% of the admissible modes. The small Galerkin errors therefore reflect a solve in nearly the whole admissible space, not a reduction of the subspace. Classical residual-greedy selection of the same size is as accurate or more accurate in every case, by up to eight orders of magnitude for the eight-qubit ellipse. Figure~\ref{fig:equal-size} compares the reconstructions at equal subspace size over the whole range of budgets. Except where both selections reach the whole sector, the hardware samples are less accurate than residual-greedy selection of the same size by factors from about 2 at the smallest budgets to about $10^{8}$ for the eight-qubit ellipse. The ideal Krylov sampler, by contrast, stays within about one order of magnitude of the greedy and oracle selections for the six-qubit ellipse, although it reaches only small subspaces with these budgets. The loss of efficiency per basis state therefore comes from the hardware noise, which adds modes with negligible solution weight.

The same noise improves coverage per shot at these sizes, because it spreads shots over modes that the ideal law rarely visits. Sample-based methods aim to reduce the subspace, however, so the relevant measure is the error per basis state, and by that measure noise is harmful. A test of dimension reduction requires admissible spaces much larger than the number of distinct samples.

In summary, the experiments on two platforms confirm the measurement convention, the parity selection rule, and the replacement-noise picture of Corollary~\ref{cor:noise}. The decay of the noiseless fraction per logical two-qubit gate is about $0.6\%$ on both platforms, whereas the gate-independent loss differs. The experiments do not demonstrate dimension reduction: the samples cover nearly the whole admissible space, and at equal subspace size classical residual-greedy selection is more accurate than the hardware samples. Larger registers are left for future work. The Trotterized lens circuits grow quickly with $T$ and are already fully randomized at eight qubits, so that more efficient compilations of the diagonal phase are needed. Clinical CT-FFR validation is a separate study that requires patient-specific geometry, physiological boundary conditions, and matched invasive measurement locations. The present cross-sectional Stokes benchmarks do not validate CT-FFR.

\section{Discussion}\label{sec:discussion}
In this method the quantum processor only proposes coordinates. The classical computer then determines their coefficients by a Galerkin solve. The bounds of Sec.~\ref{sec:galerkin-theory} relate the outcome law to the exposure of solution coordinates and turn the exposed coordinates into an energy-norm error. For positive-definite systems the Galerkin solve needs only a principal submatrix of $\hat A$, and its error decreases monotonically as the sampled set grows.

The geometric formulation of Sec.~\ref{sec:geometry} shows how these principles apply to linear systems from partial differential equations. The measurement basis matters more than the circuit. For the ducts, a Fourier basis matched to the boundary conditions reduces the subspace by one to two orders of magnitude, and the required number of modes grows slowly with the grid. The time schedule must be chosen for the linear system. The short step that is natural for eigenvalue SKQD fails on fine grids, because it does not dephase the low eigenvalues, whereas the multiscale heuristic improves on the fixed steps in several of our tests. Operator approximations are controlled by Proposition~\ref{prop:perturb}, while sampler approximations change coverage, achieved error, and postprocessing cost but not the target.

The elliptic duct and the curved vessels exhibit the mechanism by which sampling can help. When the source and the solution have different mode structures, the dephased law captures the solution while a source-based selection does not. A benefit requires that this mismatch occur together with an essential set that no classical method identifies cheaply. Code-space recovery~\cite{csr} could supply a recovery signal where native filters are absent, such as the ellipse and the curved vessels. Initial states that flatten the outcome law on the essential set are a direct target for improving the sampling efficiency $\eta_{\mathrm s}$.

The hardware experiments add two practical observations. First, the noise on both platforms behaved approximately as the replacement noise of Corollary~\ref{cor:noise}, with a strength set mainly by the number of two-qubit gates. Second, a symmetry selection rule is an inexpensive noise diagnostic: the fraction of rejected lens shots measured the noise strength directly. Such rules exist whenever the cross section has an isometry that preserves the coefficients, and they should be built into the measurement basis where possible.

\subsection{Cost of the classical selections}\label{sec:cost}
The comparisons of Secs.~\ref{sec:results} and~\ref{sec:hardware} are made at equal subspace size. They measure how much a selected set is worth, not what it costs to find it. At the sizes of this paper, with $n\leq65536$, every full vector fits in memory, and the classical selections are inexpensive. Their cost, however, grows with the full dimension $n=2^q$. The selection by $|\hat b|$ requires the transform $\hat b=Fb$ of a vector of length $n$, which costs $O(n\log n)$ operations and $O(n)$ memory. Each step of the residual-greedy selection forms the full residual $\hat b-\hat A\,\mathrm{Embed}(y,\mathcal S)$, again with $O(n\log n)$ operations and $O(n)$ memory, and the number of steps grows with $|\mathcal S|$. For $q\gtrsim45$ these vectors exceed the memory of the largest classical computers.

The quantum sampler does not form these vectors. Each shot requires the preparation of $|b\rangle$, the evolution $e^{-iAT}$, and the transform $F_q$. For the operators of Sec.~\ref{sec:geometry}, $F_q$ is a sine or cosine transform with $O(q^2)$ gates, up to the padding discussed in Sec.~\ref{sec:ducts}, and the split step of Eq.~\eqref{eq:trotter} uses only this transform and a diagonal phase. The reconstruction requires the elements $\hat A[\mathcal S,\mathcal S]$ and $\hat b[\mathcal S]$. In conformal coordinates $\hat A=\Lambda_0+\Gamma F\operatorname{diag}(h^2)F^\dagger$, and the elements of the second term are combinations of the Fourier coefficients of $h^2$ at the sums and differences of the mode indices. For a smooth metric factor these coefficients decay rapidly and can be computed by quadrature without vectors of length $n$. The reduced solve then costs at most $O(|\mathcal S|^3)$ operations. As a result, the cost of the sampler and of the reconstruction is set by $q$, the number of shots, and $|\mathcal S|$, but not by $n$.

This comparison has three qualifications. First, the preparation of $|b\rangle$ and the evolution must be efficient. For the sources and metric factors of this paper they are given by smooth functions, but efficient preparation is an assumption, not a result. Second, the residual certificate of Eq.~\eqref{eq:residual-certificate} requires the full residual, and in the large-register regime it must be replaced by an estimate. Third, a classical algorithm with sample access to $\hat b$ could select modes by $|\hat b|$ without forming full vectors. Theorem~\ref{thm:support} shows that such a selection fails when the source and the solution have different mode structures, as for the ellipse, the curved vessels, and the curved lens of Table~\ref{tab:large}. A complete comparison with classical sampling algorithms under the same input model is left for future work.

Within these qualifications, the sampler can be useful in the regime of large registers, where $n$ is beyond classical memory while $|\mathcal E_\varepsilon|$ remains moderate. Table~\ref{tab:large} shows the relevant trend: for the curved lens, $|\mathcal E_\varepsilon|$ grows by a factor of about three while $n$ grows by a factor of 64, and the error of the sampler stays near $2\times10^{-3}$ while the sampled fraction falls from 43\% to 3.2\%. The duct examples were chosen as verifiable test problems in the opposite regime, where $n$ is small and the matrices are nearly diagonal in the measurement basis.

\paragraph{Limitations} At the sizes studied, with $n\leq65536$, a classical residual-greedy selection finds the essential modes with fewer basis states than the quantum sampler in every problem, including curved vessels, chaotic cross sections, and the larger examples of Sec.~\ref{sec:large}. This comparison is made at equal subspace size and does not include the cost of forming full vectors, discussed in Sec.~\ref{sec:cost}. The gap is small for the multiscale source in the curved channel, but the selection by $|\hat b|$ is equally effective there. The measured efficiency $\eta_{\mathrm s}$ lies between $6\times10^{-5}$ and $6\times10^{-2}$, so rare modes are covered poorly. The present hardware experiments do not test the reduction of the subspace, because at these register sizes the number of shots exceeds the admissible dimension. Such a test requires larger registers and shallower circuits. Several questions also remain open. The high-probability bound of Corollary~\ref{cor:highprob-galerkin} requires complete discovery of an energy-norm essential set, which can be pessimistic; a bound in terms of retained energy would be sharper. The heuristic multiscale schedule performs well in our tests but has no guarantee. Finally, a useful-scale comparison must include classical input access, state preparation, and the full reconstruction cost, with matched accuracy.

\section{Conclusion}\label{sec:conclusion}
We have developed sample-based quantum Krylov solvers for positive-definite systems from elliptic partial differential equations. We formulated the sampling in any fixed basis with an efficient transform. For the Galerkin reconstruction we derived expected and high-probability energy-norm bounds, a sample complexity, a residual certificate, a support-completeness theorem, and a sampling efficiency. We gave a constructive procedure for curved cross sections, in which conformal coordinates fix the matrix class, the isometries of the cross section give postselection filters that keep every solution mode, and the Lie derivative of the metric bounds the effect of geometric perturbations. Numerical results for elliptic, lens-shaped, and curved ducts, and for non-separable examples with up to 16 qubits, confirm the guarantees and show strong compression in a Fourier basis. Experiments on trapped-ion and superconducting processors with six and eight qubits confirm the measurement convention, the parity selection rule as an error-detection mechanism, and the replacement-noise picture. At equal subspace size, classical selections that use full vectors are more accurate at these sizes, but their cost grows with the full dimension, whereas the cost of the sampler and of the reduced Galerkin solve is set by the number of qubits, shots, and sampled modes.

To the best of our knowledge, the analysis and the examples give the first explicit conditions under which quantum Krylov-state sampling for linear systems can be expected to offer a benefit. (1) The source and the solution must have different mode structures, so that a selection from the source fails. (2) The essential set must not be identifiable by an inexpensive classical procedure. At the sizes of this paper residual-greedy selection is such a procedure; for registers beyond classical memory its cost grows with $n$ (Sec.~\ref{sec:cost}). (3) The outcome law must be nearly flat on the essential set, so that the sampling efficiency is not small. (4) The hardware noise must be small enough that the samples do not fill the admissible space with modes of negligible solution weight. The elliptic duct and the curved vessels satisfy condition (1). The larger examples of Sec.~\ref{sec:large} show how the conditions separate. The curved lens satisfies condition (1) but not condition (2), because its solution is smooth. The multiscale channel nearly satisfies condition (2), but not condition (1). The duct examples were chosen because their exact references allow every bound to be checked, not because they satisfy these conditions. Conditions (2) and (4) define the targets for the choice of problems and for hardware improvements. Problems that may satisfy all four conditions include non-separable operators whose source and solution differ in scale, such as the forced flow between rotating cylinders. A follow-up study on IBM processors with more than 50 qubits, with problems chosen according to these conditions, is in preparation.

\section*{Acknowledgment}
This work was supported by Korea National Research Foundation (NRF) grant No. NRF-RS-2026-25612745, ICT R\&D program of IITP-2023-2021-0-01810, RS-2023-00225385, AFOSR grant FA2386-22-1-4052. The work was partly supported by the Cleveland Clinic Quantum Start-up Catalyzer Program and Amazon Web Services. Authors acknowledge the use of IBM Quantum services for this work through the Cleveland Clinic's Quantum Start-up Catalyzer Program. The views expressed are those of the authors, and do not reflect the official policy or position of IBM or the IBM Quantum team. This work was also supported by the National Quantum Laboratory at Maryland (QLab). We acknowledge QLab for providing quantum computing resources on IonQ devices for this work. The views expressed are those of the authors, and do not reflect the official policy or position of IonQ. D. Ahn is a co-founder of Singularity Quantum Inc. and holds an equity interest in the company. The other author declares no competing interests.
The authors used the AI assistants ChatGPT (OpenAI) and Claude (Anthropic) to help draft and edit the manuscript, to check the mathematical derivations, and to write numerical and analysis scripts. The authors reviewed and verified all results and take full responsibility for the content of this article.

\section*{Author contributions}
D.A. conceived the work, developed the theory and the algorithms, carried out the numerical studies, and wrote the first draft of the manuscript. B.P. carried out the hardware implementation. Both authors reviewed and edited the manuscript.

\section*{Data availability}
The derivations are given in the text and in the Supplementary Material. The hardware data, the submitted circuits, and the code that reproduces the numerical and hardware results are available from the corresponding author upon reasonable request.

\clearpage
\setcounter{section}{0}\setcounter{equation}{0}\setcounter{table}{0}\setcounter{figure}{0}
\setcounter{lemma}{0}\setcounter{proposition}{0}
\renewcommand{\thesection}{S\arabic{section}}
\renewcommand{\thesubsection}{\thesection.\arabic{subsection}}
\renewcommand{\theequation}{S\arabic{equation}}
\renewcommand{\thetable}{S\arabic{table}}
\renewcommand{\thefigure}{S\arabic{figure}}
\renewcommand{\thelemma}{S\arabic{lemma}}
\renewcommand{\theproposition}{S\arabic{proposition}}
\begin{center}{\Large\bfseries Supplementary Material}\end{center}
\vspace{1ex}
This Supplementary Material contains the proofs of the results stated in Secs.~\ref{sec:algorithm}--\ref{sec:geometry} of the main text, the analysis of finite time averaging, supporting numerical tests, and further details of the geometry, the discretization, and the hardware circuits. Equations, tables, figures, lemmas, and propositions numbered with the prefix S refer to this document. Numbers without the prefix refer to the main text. A single list of references for the main text and the Supplementary Material is given at the end.

\section{Proofs for Sections II and III}\label{sup:proofs3}
\subsection*{Proof of Lemma~\ref{lem:galerkin}}
By substituting $b=Ax^*$ into the Galerkin condition, we obtain $\langle v,A(x^*-x_{\mathcal S})\rangle=0$ for all $v\in V_{\mathcal S}$. Hence $x_{\mathcal S}$ is the $A$-orthogonal projection of $x^*$ onto $V_{\mathcal S}$. The second statement follows because $V_{\mathcal S}\subseteq V_{\mathcal S'}$.

\subsection*{Proof of Theorem~\ref{thm:error}}
Let $\Pi_{\mathcal S}$ be the orthogonal projector onto $V_{\mathcal S}$. Since $\Pi_{\mathcal S}x^*\in V_{\mathcal S}$, Lemma~\ref{lem:galerkin} gives $\|x_{\mathcal S}-x^*\|_A\leq\|\Pi_{\mathcal S}x^*-x^*\|_A$. In general $\|v\|_A^2\leq\lambda_{\max}\|v\|^2$, so that
\begin{equation}
 \|x_{\mathcal S}-x^*\|_A^2\leq\lambda_{\max}\sum_{z\notin\mathcal S}|\hat x_z|^2.
 \label{eq:tail-energy}
\end{equation}
Let $p_{k}(z)$ be the probability of index $z$ in one shot on the $k$th state. The shots are independent. By using $1-p\leq e^{-p}$, we obtain
\begin{equation}
 \begin{aligned}\Pr(z\notin\mathcal S)&=\prod_k\bigl(1-p_k(z)\bigr)^{M/d}\\&\leq\exp\Bigl(-\frac Md\sum_kp_k(z)\Bigr)=e^{-Mq(z)}.\end{aligned}
 \label{eq:miss}
\end{equation}
By substituting Eq.~\eqref{eq:miss} into the expectation of Eq.~\eqref{eq:tail-energy}, we obtain Eq.~\eqref{eq:bound}.

\subsection*{Proof of Corollary~\ref{cor:complexity}}
We split the sum in Eq.~\eqref{eq:bound} into $z\notin\mathcal E_\varepsilon$ and $z\in\mathcal E_\varepsilon$. The first part is at most $(\varepsilon^2/2)\|x^*\|_A^2$ by definition. The second part is at most $\lambda_{\max}\|x^*\|^2e^{-Mq_{\min}}$. Since $\|x^*\|_A^2\geq\lambda_{\min}\|x^*\|^2$, it is at most $\kappa e^{-Mq_{\min}}\|x^*\|_A^2$. Equation~\eqref{eq:M} makes it at most $(\varepsilon^2/2)\|x^*\|_A^2$.

\subsection*{Proof of Corollary~\ref{cor:highprob-galerkin}}
By Eq.~\eqref{eq:miss}, $\Pr(z\notin\mathcal S)\leq e^{-Mq(z)}\leq\delta/|\mathcal E_\varepsilon|$ for every $z\in\mathcal E_\varepsilon$. The union bound gives $\Pr(\mathcal E_\varepsilon\not\subseteq\mathcal S)\leq\delta$. If $\mathcal E_\varepsilon\subseteq\mathcal S$, Eq.~\eqref{eq:tail-energy} and the tail condition give the error bound.

\subsection*{Proof of Theorem~\ref{thm:support}}
By Eq.~\eqref{eq:a} and the Cauchy--Schwarz inequality, $|\hat x_z|^2=|\sum_\lambda a_\lambda(z)\lambda^{-1}|^2\leq\sum_\lambda|a_\lambda(z)|^2\sum_\lambda\lambda^{-2}$. The sum over distinct eigenvalues is at most the trace of $A^{-2}$.

\section{Finite time averaging and schedule mixtures}\label{sup:dephasing}
For uniform times $T_k=k\Delta t$, $k=0,\dots,d-1$, let $D_d(\theta)=d^{-1}\sum_{k=0}^{d-1}e^{-ik\theta}$. Expanding the amplitudes in the eigenbasis gives the following lemma.

\begin{lemma}[Finite averaging]\label{lem:finite}
\begin{equation}
 \begin{gathered}|q(z)-q^\infty(z)|\leq\sum_{\lambda\neq\mu}|a_\lambda(z)||a_\mu(z)|\,\bigl|D_d\bigl((\lambda-\mu)\Delta t\bigr)\bigr|,\\
 |D_d(\theta)|\leq\min\Bigl\{1,\frac{1}{d|\sin(\theta/2)|}\Bigr\}.\end{gathered}
 \label{eq:dirichlet}
\end{equation}
\end{lemma}

Equation~\eqref{eq:dirichlet} shows that dephasing requires $d|\sin((\lambda-\mu)\Delta t/2)|\gg1$ for the eigenvalue pairs that carry the solution. The required window depends on the relevant eigenvalue differences and on aliasing, not on $\lambda_{\min}$ alone, and near-degenerate pairs may require windows much longer than $2\pi/\lambda_{\min}$. The step $\Delta t=\pi/\lambda_{\max}$, which we use as the reference short step, gives a total time that is shorter than $2\pi/\lambda_{\min}$ by a factor of about $2\kappa/d$. A single grid can fail completely: if $(\lambda-\mu)\Delta t$ is a multiple of $2\pi$ for every eigenvalue pair with $a_\lambda(z)a_\mu(z)\neq0$, then $D_d=1$ for these pairs and $q(z)=|\hat b_z|^2$, so a mode hidden by cancellation in $\hat b$ stays hidden. Incommensurate or varied time steps avoid this failure.

\subsection*{Proof of Lemma~\ref{lem:finite}}
By writing $|b\rangle=\sum_\lambda P_\lambda|b\rangle$, we obtain $\langle z|_Fe^{-iAk\Delta t}|b\rangle=\sum_\lambda a_\lambda(z)e^{-i\lambda k\Delta t}$. Hence
\begin{equation}
 \begin{aligned}q(z)&=\frac1d\sum_{k=0}^{d-1}\Bigl|\sum_\lambda a_\lambda(z)e^{-i\lambda k\Delta t}\Bigr|^2\\&=\sum_{\lambda,\mu}a_\lambda(z)\overline{a_\mu(z)}\,D_d\bigl((\lambda-\mu)\Delta t\bigr).\end{aligned}
\end{equation}
The terms with $\lambda=\mu$ give $q^\infty(z)$, since $D_d(0)=1$. The remaining terms give Eq.~\eqref{eq:dirichlet}. The geometric sum gives $|D_d(\theta)|=|\sin(d\theta/2)|/(d|\sin(\theta/2)|)\leq1/(d|\sin(\theta/2)|)$, and $|D_d|\leq1$ holds trivially.

\subsection{Schedule mixtures}
A long uniform step aliases the high eigenvalues, while a short step does not dephase the low ones. An explicit mixture of candidate schedules has the following coverage property.

\begin{lemma}[Schedule mixtures]\label{lem:hedging}
Let $\mathcal T_1,\dots,\mathcal T_r$ be complete time schedules with outcome laws $q^{(1)},\dots,q^{(r)}$, and allocate equal total shots to each schedule. The outcome law of the mixture satisfies
\[
 q^{\mathrm{mix}}=\frac1r\sum_{j=1}^rq^{(j)}\geq\frac{q^{(j)}}{r}\qquad\text{for every }j.
\]
\end{lemma}

By Corollary~\ref{cor:complexity}, the sufficient budget of the mixture is at most $r$ times that of its best constituent for the same essential set. This compares sufficient bounds, not observed stopping times. Times that occur in several constituents are kept with multiplicity. Algorithm~\ref{alg:times} instead selects isolated, logarithmically spaced times. It does not contain the candidate fixed-step schedules, and Lemma~\ref{lem:hedging} does not apply to it. We use it as an empirical heuristic. Its lower time scale $\pi/\lambda_{\max}$ avoids aliasing of the highest eigenvalues, and its upper scale $2\pi/\lambda_{\min}$ resolves the lowest eigenvalues, which dominate $\hat x$ through the weights $1/\lambda$ of Eq.~\eqref{eq:a}.

\section{Proofs for Section IV}\label{sup:proofs4}
\subsection*{Proof of Proposition~\ref{prop:conformal}}
In two dimensions $\sqrt g=h^2$ and $g^{ij}=h^{-2}\delta^{ij}$, so their product is $\delta^{ij}$.

\subsection*{Proof of Proposition~\ref{prop:symmetry}}
(i) Let $\phi_s$ be the flow of $X$. Since $\phi_s$ is an isometry that maps $\Omega$ and its wall onto themselves and preserves $\varrho$ and $\gamma$, the bilinear form of Eq.~\eqref{eq:weak} satisfies $a(\phi_s^*w,\phi_s^*v)=a(w,v)$, and similarly for the load. Hence $\phi_s^*w$ is the solution for the load $\phi_s^*f$. By differentiating at $s=0$, we obtain that $\mathcal L_Xw$ is the solution for the load $\mathcal L_Xf$. The two statements follow from the linearity and uniqueness of the solution.
(ii) Since $P$ commutes with $A$, it commutes with $e^{-iAT}$ and with $A^{-1}$. Hence $Pe^{-iAT}b=\chi e^{-iAT}b$ and $Px^*=\chi x^*$. For a vector $v$ with $Pv=\chi v$, we have $\chi_z(Fv)_z=(FPF^\dagger Fv)_z=\chi(Fv)_z$, so $(Fv)_z=0$ when $\chi_z\neq\chi$.

\subsection*{Proof of Proposition~\ref{prop:perturb}}
Let $\mathcal E=A^{-1/2}\,\delta A\,A^{-1/2}$, so that $\|\mathcal E\|\leq\varepsilon$. Since $\tilde x-x^*=-\tilde A^{-1}\delta A\,x^*+\tilde A^{-1}\delta b$, we obtain $A^{1/2}(\tilde x-x^*)=-(I+\mathcal E)^{-1}\mathcal E\,A^{1/2}x^*+(I+\mathcal E)^{-1}A^{-1/2}\delta b$. Hence $\|\tilde x-x^*\|_A\leq(1-\varepsilon)^{-1}(\varepsilon\|x^*\|_A+\|\delta b\|_{A^{-1}})$. The second statement follows because $0\preceq C\preceq A$ and $0\preceq K\preceq A$.

\subsection*{Proof of Proposition~\ref{prop:shape}}
We use $\delta g^{ij}=-g^{ik}\delta g_{kl}g^{lj}$ and $\delta\sqrt g=\frac12\sqrt g\,g^{kl}\delta g_{kl}$ with $\delta g=s\,\mathcal L_Vg$. Since $(\mathcal L_Vg)_{ij}=\nabla_iV_j+\nabla_jV_i$, we have $\frac12g^{kl}(\mathcal L_Vg)_{kl}=\operatorname{div}_gV$. By substituting these expressions into $\delta(\sqrt g\,g^{ij})=\delta\sqrt g\,g^{ij}+\sqrt g\,\delta g^{ij}$, we obtain Eq.~\eqref{eq:shape}. Relative to $g$, the trace-free part has the eigenvalues $\pm\frac12(\mu_1-\mu_2)$, and $\operatorname{div}_gV=\frac12(\mu_1+\mu_2)$. The relative changes of the blocks of $W$ and of the entries of $C$ are therefore bounded by the two terms of Eq.~\eqref{eq:shape-eps}, and the bound on $\langle v,\delta A\,v\rangle$ follows as in the proof of Proposition~\ref{prop:perturb}. For $\mathcal L_Vg=2\sigma g$ the trace-free part vanishes.

\section{Elements of Riemannian geometry}\label{app:riemann}
This section summarizes the notions of Riemannian geometry used in Sec.~\ref{sec:geometry}. Standard references are Refs.~\cite{Nakahara2003,Jost2017,Chavel1984}. We use the summation convention.

\paragraph{Metric and volume} A Riemannian manifold $(\Omega,g)$ is a smooth manifold with a positive-definite symmetric tensor $g_{ij}$ in local coordinates $u^i$. The line element is $ds^2=g_{ij}du^idu^j$, and $g^{ij}$ denotes the inverse matrix. The volume element is $d\mu_g=\sqrt g\,du^1\cdots du^n$ with $\sqrt g=\sqrt{\det g_{ij}}$. For a map $\Phi:R\to\Omega\subset\mathbb R^n$ from a coordinate domain, the induced metric is $g_{ij}=\partial_i\Phi\cdot\partial_j\Phi$.

\paragraph{Gradient, divergence, and the Laplace--Beltrami operator} The gradient of a function is $(\nabla_gw)^i=g^{ij}\partial_jw$, and the divergence of a vector field is $\operatorname{div}_gV=g^{-1/2}\partial_i(\sqrt g\,V^i)$. The weighted operator of Eq.~\eqref{eq:weak} is
\begin{equation}
 -\nabla_g\cdot(\varrho\nabla_gw)=-\frac1{\sqrt g}\,\partial_i\bigl(\varrho\sqrt g\,g^{ij}\partial_jw\bigr).
 \label{eq:laplace-beltrami}
\end{equation}
For $\varrho=1$ it is the Laplace--Beltrami operator $-\Delta_g$. By multiplying Eq.~\eqref{eq:laplace-beltrami} by a test function $v$ that vanishes on the wall and integrating with $d\mu_g$, we obtain the left side of Eq.~\eqref{eq:weak}. The integrand $\varrho\,g^{ij}\partial_iw\,\partial_jv\,\sqrt g$ is a density, so the weak form does not depend on the choice of coordinates.

\paragraph{Conformal metrics in two dimensions} If $g_{ij}=h^2\delta_{ij}$, then $\sqrt g=h^2$ and $\sqrt g\,g^{ij}=\delta^{ij}$, so that $\Delta_g=h^{-2}\Delta$ and the Dirichlet energy $\int g^{ij}\partial_iw\,\partial_jw\,d\mu_g$ equals its flat value. This is the content of Proposition~\ref{prop:conformal}. The Gaussian curvature of a conformal metric is $\mathcal K=-h^{-2}\Delta\ln h$. A metric induced by a conformal map of a planar domain is flat, so $\ln h$ is harmonic. For the elliptic coordinates, $h^2=\ell_g^2|\sinh(\mu+i\nu)|^2$, and for the bipolar coordinates, $h=\ell_g/|\cosh\tau-\cos\sigma|$; in both cases $\ln h$ is the real part of the logarithm of an analytic function. This identity provides a simple check of a numerically computed metric factor.

\paragraph{Lie derivative} Let $X$ be a vector field with flow $\phi_s$. The Lie derivative of a tensor field $T$ is
\begin{equation}
 \mathcal L_XT=\frac{d}{ds}\Big|_{s=0}\phi_s^*T .
 \label{eq:lie}
\end{equation}
For a function, $\mathcal L_Xw=X^i\partial_iw$. For a vector field, $\mathcal L_XY=[X,Y]$. For the metric,
\begin{equation}
 (\mathcal L_Xg)_{ij}=X^k\partial_kg_{ij}+g_{kj}\partial_iX^k+g_{ik}\partial_jX^k=\nabla_iX_j+\nabla_jX_i ,
 \label{eq:lie-metric}
\end{equation}
where $\nabla$ is the Levi-Civita connection. The trace is $g^{ij}(\mathcal L_Xg)_{ij}=2\operatorname{div}_gX$. Pull-back commutes with $d/ds$, so $\frac{d}{ds}\phi_s^*T=\phi_s^*\mathcal L_XT$ for all $s$.

\paragraph{Killing and conformal Killing fields} A Killing field satisfies $\mathcal L_Xg=0$. Its flow consists of isometries, and it commutes with $\Delta_g$. Examples are $\partial_\vartheta$ for a disk in polar coordinates and translations along a straight duct. A conformal Killing field satisfies $\mathcal L_Xg=2\sigma g$ for a function $\sigma$. In two dimensions every analytic function $\zeta(u^1+iu^2)$ defines a conformal Killing field of the flat metric, which is why conformal maps are so plentiful in the plane. Proposition~\ref{prop:symmetry} uses Killing fields and discrete isometries, and Proposition~\ref{prop:shape} uses the trace-free part of $\mathcal L_Vg$, which vanishes exactly for conformal Killing fields.

\paragraph{Variations of the metric} For a variation $\delta g_{ij}$,
\begin{equation}
 \delta g^{ij}=-g^{ik}\delta g_{kl}g^{lj},\qquad \delta\sqrt g=\tfrac12\sqrt g\,g^{kl}\delta g_{kl}.
 \label{eq:metric-variation}
\end{equation}
With $\delta g=s\mathcal L_Vg$ these formulas give Eq.~\eqref{eq:shape}. A moving wall, as in a pulsating vessel, can be treated in the same way: if the wall moves with velocity $V$, the pulled-back metric evolves as $\partial_tg=\mathcal L_Vg$, which is the geometric content of an arbitrary Lagrangian--Eulerian formulation.

\paragraph{Tube coordinates} Let a centerline have arc length $s$, curvature $\kappa_c$, and torsion $\tau_c$. Let $(\rho,\vartheta)$ be polar coordinates in the normal plane, with $\vartheta$ measured from a direction that rotates relative to the principal normal by the accumulated torsion angle $\vartheta_0(s)$, $|\vartheta_0'|=|\tau_c|$. Germano showed that this choice removes the off-diagonal terms and gives the orthogonal metric $ds^2=d\rho^2+\rho^2d\vartheta^2+H^2ds^2$, where $H=1-\kappa_c\rho\cos(\vartheta-\vartheta_0)$~\cite{Germano1982}. Over a short segment, or for $\tau_c=0$, $\vartheta_0$ can be set to zero, which gives the factor $H$ of Sec.~\ref{sec:curved}. The cross-sectional coefficients of Eq.~\eqref{eq:curved} follow from this metric.

\paragraph{Spectral asymptotics} For the Dirichlet Laplace--Beltrami operator on a two-dimensional domain, Weyl's law states that the number of eigenvalues below $\lambda$ grows as $|\Omega|_g\lambda/(4\pi)$~\cite{Chavel1984}. Since $|\Omega|_g=\int\sqrt g$ is unchanged by a change of coordinates, the high-frequency part of the spectrum does not depend on the coordinate map.

\section{Verification of the energy-norm theory}\label{sup:verification}
\paragraph{Support completeness} Table~\ref{tab:support} tests Eq.~\eqref{eq:support} without a support threshold in the bound itself. The ratio $q^\infty(z)\operatorname{tr}A^{-2}/|\hat x_z|^2$ must be at least 1 for every mode, and it is. We compute $q^\infty$ with degenerate eigenvalues grouped into spectral projectors (tolerance $10^{-6}\lambda_{\max}$). Without this grouping, exactly degenerate pairs from different symmetry sectors are split and produce spurious weight outside the sector. The same ratio formed with the single-state law $|\hat b_z|^2$ vanishes for the ellipse. The source $h^2=\ell_g^2(\sinh^2\mu+\sin^2\nu)$ excites exactly 32 or 64 modes, whereas the solution has weight above $10^{-12}$ in 93 to 139 modes. For the lens, both $\hat b$ and $\hat x$ lie in the odd--odd sector of dimension $N^2/4$. The solution coefficients decay rapidly, so the number of modes counted as nonzero depends on the threshold; Table~\ref{tab:support} gives two thresholds. The minimum ratios do not change between them. Table~\ref{tab:scaling} shows the consequence of the ellipse mismatch: a selection by $|\hat b|$ matches the oracle for the lens and fails for the ellipse.

\begin{table*}[t]
\centering
\caption{Support completeness, Eq.~\eqref{eq:support}. Minimum over modes with $|\hat x_z|^2\geq10^{-24}\|\hat x\|^2$ of the dephased ratio $q^\infty(z)\operatorname{tr}A^{-2}/|\hat x_z|^2$ (bound: $\geq1$) and of the single-state ratio $|\hat b_z|^2\operatorname{tr}A^{-2}/|\hat x_z|^2$. The last three columns count solution modes with relative weight above $10^{-12}$ and $10^{-24}$, and the exact support of $\hat b$.}\label{tab:support}
\begin{tabular}{lccccc}
\toprule
Case & dephased ratio & single-state ratio & $\#\hat x\,(10^{-12})$ & $\#\hat x\,(10^{-24})$ & $|\operatorname{supp}\hat b|$\\
\midrule
Ellipse $16^2$, $\Gamma=1000$ & 10.7 & 0 & 93 & 236 & 32\\
Ellipse $32^2$, $\Gamma=100$ & 4.0 & 0 & 94 & 266 & 64\\
Ellipse $32^2$, $\Gamma=1000$ & 9.4 & 0 & 139 & 416 & 64\\
Lens $16^2$, $\Gamma=1000$ & 791 & $3.8\times10^{3}$ & 64 & 64 & 64\\
Lens $32^2$, $\Gamma=100$ & 11.9 & 3.4 & 150 & 256 & 256\\
Lens $32^2$, $\Gamma=1000$ & 435 & 3.2 & 167 & 256 & 256\\
\bottomrule
\end{tabular}
\end{table*}

\paragraph{Error bound} Table~\ref{tab:bound} compares the observed errors with Theorem~\ref{thm:error}. The bound holds in all cases. It exceeds the observed root-mean-square error by a factor of 1.2 to 18, and by less than 4 in ten of the twelve entries.

\begin{table*}[t]
\centering
\caption{Relative energy-norm error: observed root-mean-square value over four sampling runs / square root of the bound of Theorem~\ref{thm:error}.}\label{tab:bound}
\begin{tabular}{llcc}
\toprule
Case & times & $M=4800$ & $M=19200$\\
\midrule
Ellipse $16^2$, $\Gamma=1000$ & $k\pi/\lambda_{\max}$ & $2.7\times10^{-4}$ / $2.8\times10^{-3}$ & $1.1\times10^{-4}$ / $2.6\times10^{-4}$\\
Ellipse $32^2$, $\Gamma=100$ & $k\pi/\sqrt{\lambda_{\min}\lambda_{\max}}$ & $4.6\times10^{-3}$ / $7.2\times10^{-3}$ & $3.7\times10^{-5}$ / $6.5\times10^{-4}$\\
Ellipse $32^2$, $\Gamma=1000$ & $k\pi/\lambda_{\max}$ & $1.1\times10^{-2}$ / $2.1\times10^{-2}$ & $2.7\times10^{-3}$ / $1.0\times10^{-2}$\\
Lens $16^2$, $\Gamma=1000$ & $k\pi/\lambda_{\max}$ & $3.6\times10^{-5}$ / $1.4\times10^{-4}$ & $8.7\times10^{-6}$ / $1.3\times10^{-5}$\\
Lens $32^2$, $\Gamma=100$ & $k\pi/\lambda_{\max}$ & $8.7\times10^{-3}$ / $1.3\times10^{-2}$ & $3.2\times10^{-3}$ / $4.8\times10^{-3}$\\
Lens $32^2$, $\Gamma=1000$ & $k\pi/\lambda_{\max}$ & $1.2\times10^{-2}$ / $1.8\times10^{-2}$ & $4.8\times10^{-3}$ / $5.7\times10^{-3}$\\
\bottomrule
\end{tabular}
\end{table*}

\paragraph{Uniform steps and the multiscale heuristic} Table~\ref{tab:times} lists $q_{\min}(\mathcal E_\varepsilon)$ for two uniform steps and for Algorithm~\ref{alg:times}. The short step $\pi/\lambda_{\max}$ fails on the finer elliptic grids, where $q_{\min}$ drops to $10^{-15}$--$10^{-20}$, in agreement with Lemma~\ref{lem:finite}. The geometric-mean step is better but loses a factor of 37 on the $64\times64$ ellipse. The multiscale heuristic is within a factor of 2 of the better uniform step in all twelve cases and is the best choice in six. This is an empirical observation; Lemma~\ref{lem:hedging} does not apply to Algorithm~\ref{alg:times}.

\begin{table*}[t]
\centering
\caption{Smallest sampling probability $q_{\min}(\mathcal E_\varepsilon)$ on the essential set ($\varepsilon=1\%$) for three sets of $d=12$ times. The last column is the ratio of the multiscale value to the better of the two uniform steps.}\label{tab:times}
\begin{tabular}{lrcccc}
\toprule
Case & $|\mathcal E_\varepsilon|$ & $k\pi/\lambda_{\max}$ & $k\pi/\sqrt{\lambda_{\min}\lambda_{\max}}$ & multiscale & ratio\\
\midrule
Ellipse $16^2$, $\Gamma=100$ & 37 & $1.2\times10^{-7}$ & $6.7\times10^{-4}$ & $6.5\times10^{-4}$ & 0.97\\
Ellipse $16^2$, $\Gamma=1000$ & 41 & $1.3\times10^{-3}$ & $2.6\times10^{-4}$ & $7.2\times10^{-4}$ & 0.55\\
Ellipse $32^2$, $\Gamma=100$ & 56 & $5.3\times10^{-15}$ & $2.9\times10^{-4}$ & $3.0\times10^{-4}$ & 1.03\\
Ellipse $32^2$, $\Gamma=1000$ & 75 & $2.8\times10^{-6}$ & $9.3\times10^{-5}$ & $1.5\times10^{-4}$ & 1.6\\
Ellipse $64^2$, $\Gamma=100$ & 67 & $8.4\times10^{-20}$ & $4.1\times10^{-6}$ & $1.5\times10^{-4}$ & 37\\
Ellipse $64^2$, $\Gamma=1000$ & 118 & $5.0\times10^{-15}$ & $3.8\times10^{-5}$ & $7.5\times10^{-5}$ & 2.0\\
Lens $16^2$, $\Gamma=100$ & 35 & $2.1\times10^{-4}$ & $1.2\times10^{-4}$ & $1.6\times10^{-4}$ & 0.76\\
Lens $16^2$, $\Gamma=1000$ & 31 & $2.0\times10^{-3}$ & $2.1\times10^{-3}$ & $2.0\times10^{-3}$ & 0.95\\
Lens $32^2$, $\Gamma=100$ & 58 & $5.0\times10^{-5}$ & $5.4\times10^{-5}$ & $1.2\times10^{-4}$ & 2.2\\
Lens $32^2$, $\Gamma=1000$ & 59 & $4.4\times10^{-5}$ & $3.9\times10^{-5}$ & $4.3\times10^{-5}$ & 0.98\\
Lens $64^2$, $\Gamma=100$ & 90 & $5.7\times10^{-6}$ & $7.3\times10^{-5}$ & $8.8\times10^{-5}$ & 1.2\\
Lens $64^2$, $\Gamma=1000$ & 107 & $3.8\times10^{-5}$ & $8.8\times10^{-6}$ & $2.0\times10^{-5}$ & 0.53\\
\bottomrule
\end{tabular}
\end{table*}

\paragraph{Efficiency and initial state} With the multiscale heuristic, the efficiency of Eq.~\eqref{eq:eta} lies between $2\times10^{-3}$ and $6\times10^{-2}$. The sampler is therefore 16 to 480 times less efficient than a uniform law on $\mathcal E_\varepsilon$. It seems natural to start the evolution from an approximate solution, such as $|\Lambda_0^{-1}b\rangle$, so that the sampling resembles $|\hat x|^2$. This was harmful for the tested initial states. In the four cases tested, $q_{\min}(\mathcal E_\varepsilon)$ drops by four to six orders of magnitude and the error at $M=19200$ increases to about $10\%$. Proposition~\ref{prop:eff} explains the failure: a law close to $|\hat x|^2$ is strongly concentrated, while efficient coverage requires a flat law on $\mathcal E_\varepsilon$.

\section{Supplementary analyses}\label{app:supp}
This section collects analyses that support the main text but are not needed to follow it.

\subsection{Green's-function interpretation}\label{sec:green}
The sampling laws have a natural interpretation in terms of Green's functions~\cite{KadanoffBaym1962}. We regard $A$ as a single-particle Hamiltonian and define, for the basis vector $|z\rangle_F$ and the source $b$,
\begin{equation}
\begin{gathered}
 G^{\mathrm{prop}}_{zb}(T)=-i\langle z|_Fe^{-iAT}|b\rangle,\\
 G^{R}_{zb}(\omega)=\langle z|_F(\omega-A+i0^+)^{-1}|b\rangle,\\
 \rho_{zb}(\omega)=2\pi\sum_\lambda a_\lambda(z)\,\delta(\omega-\lambda).
\end{gathered}
 \label{eq:green}
\end{equation}
Here $\rho_{zb}$ is a cross spectral measure and $G^{\mathrm{prop}}$ is a propagator matrix element. It is not a many-body greater Green's function, which would require a specified state, field operators, and occupation factors. The resolvent interpretation is exact for the finite matrix without those statistical assumptions.

\begin{proposition}[Green's-function form]\label{prop:green}
The quantities of Sec.~\ref{sec:galerkin-theory} are
\begin{align}
 q(z)&=\frac1d\sum_k\bigl|G^{\mathrm{prop}}_{zb}(T_k)\bigr|^2,\qquad
 \nonumber\\
 q^\infty(z)&=\lim_{\mathcal T\to\infty}\frac1{\mathcal T}\int_0^{\mathcal T}\bigl|G^{\mathrm{prop}}_{zb}(T)\bigr|^2dT,\label{eq:green1}\\
 \hat b_z&=\int\frac{d\omega}{2\pi}\rho_{zb}(\omega),\qquad
 \nonumber\\
 \hat x_z&=-G^R_{zb}(0)=\int\frac{d\omega}{2\pi}\frac{\rho_{zb}(\omega)}{\omega}.\label{eq:green2}
\end{align}
\end{proposition}
\begin{proof}
The first relation is Eq.~\eqref{eq:q}. The second follows from $G^{\mathrm{prop}}_{zb}(T)=-i\sum_\lambda a_\lambda(z)e^{-i\lambda T}$, since the cross terms with $\lambda\neq\mu$ average to zero. Equation~\eqref{eq:green2} follows from Eq.~\eqref{eq:a}.
\end{proof}

Proposition~\ref{prop:green} gives the sampling theory a spectral meaning. The source is the zeroth moment of the cross spectral measure, equal to $\langle z|_Fb\rangle$, rather than a unit-normalized diagonal spectral sum rule; cross spectral weights can have either sign or be complex. The solution is the static response, that is, the moment of order $-1$. A single state measures only the zeroth moment, which can vanish by cancellation. The dephased law $q^\infty$ is the sum of squared cross amplitudes at distinct eigenvalues. It does not measure the signed spectral measure itself. A time window of length $\mathcal T$ resolves frequencies only to about $2\pi/\mathcal T$, and the moment of order $-1$ is dominated by the lowest poles. This motivates the upper time scale of Algorithm~\ref{alg:times}; actual dephasing requires resolution of the relevant eigenvalue differences, as in Lemma~\ref{lem:finite}. On the other hand, the quantum device returns only intensities $|G^{\mathrm{prop}}_{zb}|^2$. The method therefore locates the spectral weight but does not reconstruct the Green's function, in contrast with the approach of Ref.~\cite{Patel2026}.

\subsection{Spectral geometry and sampling efficiency}\label{sec:ergodic}
The efficiency $\eta_{\mathrm s}$ of Eq.~\eqref{eq:eta} depends on how the eigenvectors of $A$ are spread over the measurement basis, since $q^\infty(z)=\sum_\lambda|a_\lambda(z)|^2$. In spectral geometry, eigenfunctions of domains with integrable geodesic flow, such as the disk and the ellipse, are structured, while those of chaotic domains, such as the stadium, become equidistributed at high energy~\cite{Shnirelman1974,Zelditch1987,ColinDeVerdiere1985,Bunimovich1979}. This suggests a hypothesis: a chaotic cross section might give a flatter $q$, a larger $\eta_{\mathrm s}$, and an essential set that is harder to find classically. We tested this hypothesis in Sec.~\ref{sec:geomresults} and found no support for it. A theoretical reason is that quantum ergodicity is a statement about high-lying eigenfunctions, while the solution of an elliptic problem with smooth data is carried by the lowest modes.

\paragraph{Numerical test} We tested the hypothesis of Sec.~\ref{sec:ergodic} in two ways. First, we used smooth non-separable conformal deformations of a square, $h^2=\exp(2\,\mathrm{Re}\,\phi)$ with $\phi$ a random cubic polynomial of amplitude between 0.3 and 1.5. Compared with a separable metric factor on the same square, the efficiency did not increase ($\eta_{\mathrm s}=3\times10^{-4}$ to $3\times10^{-3}$ versus $8\times10^{-4}$), and the residual-greedy selection remained 2 to 13 times more accurate than the Krylov sampling. Second, we compared cross sections of equal area with integrable and chaotic geodesic flow by Brinkman penalization on a bounding box (Table~\ref{tab:shapes}). Penalization is a Class I approximation that avoids a conformal map. The sharp interface spreads the solution in the Fourier basis, so the subspaces are larger and the errors higher than with conformal coordinates. The chaotic stadium does not give a higher efficiency than the disk, the ellipse, or the square. In all four cases the residual-greedy selection is about twice as accurate as the Krylov sampling.

\begin{table*}[t]
\centering
\small\setlength{\tabcolsep}{4pt}
\caption{Cross sections of equal area by Brinkman penalization ($32\times32$ bounding box, $\Gamma=100$, penalty $10^4$, multiscale heuristic, $M=19200$). IPR is the mean inverse participation ratio of the 40 lowest eigenvectors in the Fourier basis.}\label{tab:shapes}
\resizebox{\linewidth}{!}{\begin{tabular}{lrcccccc}
\toprule
Cross section & $|\mathcal E_\varepsilon|$ & $\eta_{\mathrm s}$ & IPR & $|\mathcal S|$ & Krylov & top-$|\hat b|$ & residual-greedy\\
\midrule
Disk (integr.) & 249 & $4.6\times10^{-4}$ & 0.094 & 198 & $1.2\times10^{-1}$ & $1.2\times10^{-1}$ & $5.4\times10^{-2}$\\
Ellipse $B_e/A_e=0.6$ (integr.) & 247 & $4.5\times10^{-4}$ & 0.165 & 186 & $1.3\times10^{-1}$ & $1.8\times10^{-1}$ & $5.6\times10^{-2}$\\
Square (integr.) & 225 & $5.9\times10^{-5}$ & 0.147 & 117 & $9.3\times10^{-2}$ & $8.5\times10^{-2}$ & $3.7\times10^{-2}$\\
Stadium (chaotic) & 240 & $1.3\times10^{-4}$ & 0.173 & 188 & $1.1\times10^{-1}$ & $9.2\times10^{-2}$ & $5.4\times10^{-2}$\\
\bottomrule
\end{tabular}}
\end{table*}

\subsection{Metric perturbation}\label{app:metric}
We perturbed the metric factor of the lens and the disk by $\delta(h^2)=\varepsilon h^2\sin(k_1\sigma+k_2\tau+\phi)$ with random $k_1,k_2,\phi$, which models the error of a numerical conformal map, with the load held fixed. Table~\ref{tab:perturb} shows the largest relative energy-norm change of the solution over five samples. The bound of Proposition~\ref{prop:perturb} holds and is within a factor of 1.7 of the observed change.

\begin{table*}[t]
\centering
\caption{Class I approximation: effect of a relative metric error $\varepsilon$ ($32\times32$, $\Gamma=100$, fixed load). Largest observed relative energy-norm change over five samples versus the bound $\varepsilon/(1-\varepsilon)$.}\label{tab:perturb}
\begin{tabular}{cccc}
\toprule
$\varepsilon$ & Lens $120^\circ$ & Disk & Bound\\
\midrule
$10^{-3}$ & $6.8\times10^{-4}$ & $6.0\times10^{-4}$ & $1.0\times10^{-3}$\\
$10^{-2}$ & $5.5\times10^{-3}$ & $6.4\times10^{-3}$ & $1.0\times10^{-2}$\\
$5\times10^{-2}$ & $3.3\times10^{-2}$ & $3.0\times10^{-2}$ & $5.3\times10^{-2}$\\
\bottomrule
\end{tabular}
\end{table*}

\subsection{Circuit resources and source preparation}\label{app:resources}
Resource accounting must include preparation, evolution, and the final basis transform. In the ellipse, the separable source has exact Schmidt rank two and the evolution factors into two register unitaries. Generic synthesis of a register unitary costs exponentially in the register size. Separability alone therefore does not give an efficient algorithm asymptotically.

For the lens, the source is generally not exactly rank two. A Schmidt truncation with retained rank $r_b$ has relative amplitude tail
$\epsilon_b=(\sum_{j>r_b}s_j^2/\sum_js_j^2)^{1/2}$, where $s_j$ are its singular values. A truncated preparation must report this tail and its effect on the sampling law. The pilot uses exact state preparation and exact DST-I register unitaries. Resource estimates that assume a power-of-two QFT or a rank-two source do not apply to it.

A compiled resource table should state the software versions, synthesis seed, time schedule, Trotter refinement, gate set, number of active qubits, and whether the counts are logical or native to the backend. We do not report lens gate counts here, because they require a construction consistent with this grid and preparation, and we make no claim of hardware feasibility.

\subsection{IonQ outcome laws}\label{app:ionq-laws}
Figure~\ref{fig:ionq} shows the circuit laws and the measured IonQ frequencies for the six six-qubit circuits of Sec.~\ref{sec:hardware}.

\begin{figure*}[t]
\centering
\includegraphics[width=\textwidth]{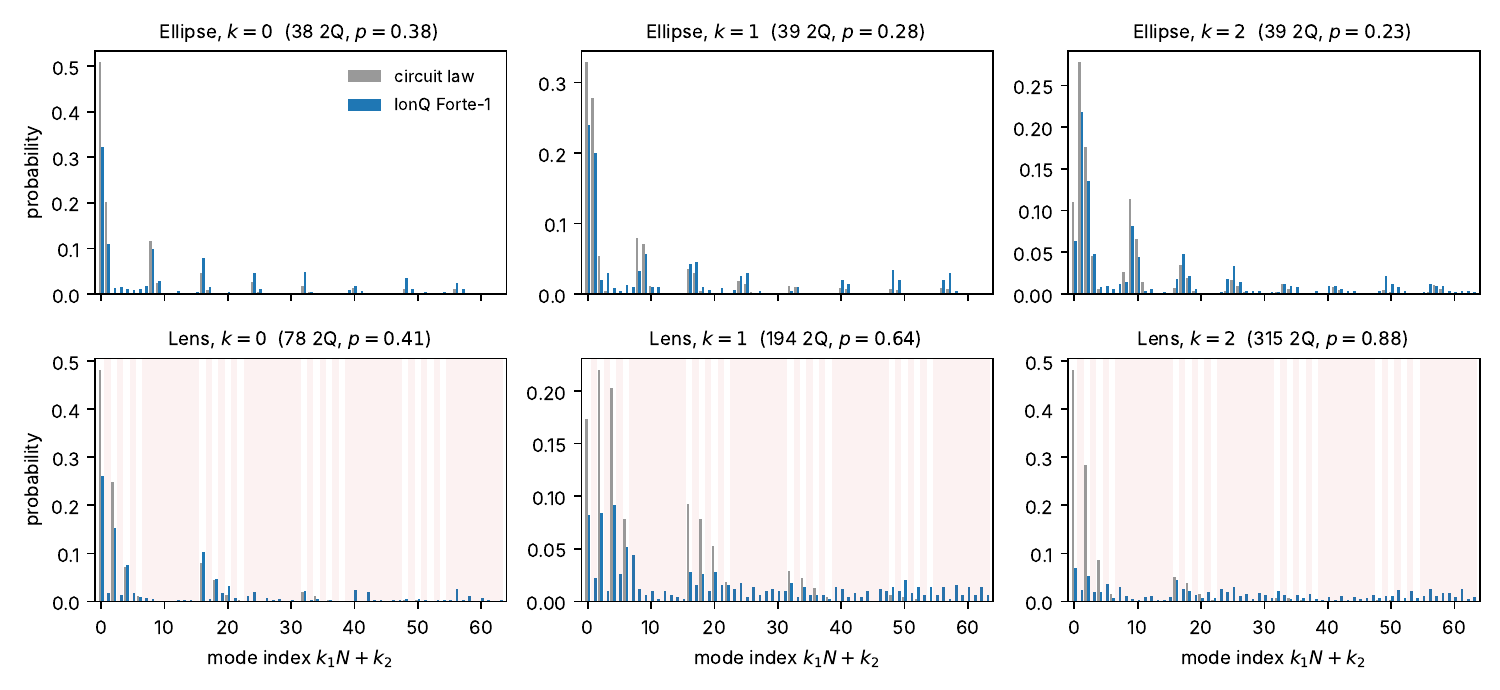}
\caption{Circuit laws (gray) and measured frequencies on IonQ Forte-1 (blue, 500 shots) for the six circuits. The titles give the number of two-qubit gates and the fitted replacement probability $p$. For the lens, the shaded indices lie outside the odd--odd sector and are removed by the parity filter.}\label{fig:ionq}
\end{figure*}

\section{Discretization of the duct examples}\label{app:discretization}
For the elliptic duct we use the cell centers $\mu_i=(i+\tfrac12)\Delta\mu$ and $\nu_j=(j+\tfrac12)\Delta\nu$ with $\Delta\mu=\mu_0/N$ and $\Delta\nu=\pi/(2N)$. The Neumann condition at $\mu=0$ is imposed by the mirror ghost value $w_{-1}=w_0$, and the Dirichlet condition at $\mu_0$ by the antisymmetric ghost value $w_N=-w_{N-1}$. The resulting tridiagonal matrix has the eigenvectors $\cos(\pi(k+\tfrac12)(i+\tfrac12)/N)$, which form the DCT-IV basis. With Neumann conditions at both ends, the eigenvectors are $\cos(\pi k(j+\tfrac12)/N)$, which form the DCT-II basis. For the lens we use the interior points $\sigma_i=\sigma_0+i\Delta\sigma$ and $\tau_j=-\tau_{\max}+j\Delta\tau$, $i,j=1,\dots,N$, with Dirichlet conditions, for which the DST-I diagonalizes the one-dimensional Laplacians. The flow rate is $Q=\iint w\,h^2\,du\,dv$, evaluated by the midpoint rule. For the curved model~\eqref{eq:curved} the face coefficients $\varrho$ are evaluated at the cell faces, which keeps $K$ symmetric. For row-major mode index $k_1N+k_2$ with $N=2^m$, the second coordinate occupies qubits $0,\ldots,m-1$ and the first occupies qubits $m,\ldots,2m-1$. Thus qubit 0 is the least significant bit of the second coordinate. Both exported circuits and returned integer outcomes must use this convention.

\section{Spectral bounds from the geometry}\label{sec:spectral-bounds}
The time schedules of Sec.~\ref{sec:support} and the residual certificate of Eq.~\eqref{eq:residual-certificate} require the bounds $a$ and $B$ of Eq.~\eqref{eq:spectrum}. The following proposition supplies them from the geometry alone.

\begin{proposition}[Spectral bounds]\label{prop:spectral}
Let $C$ be diagonal with entries in $[c_{\min},c_{\max}]$. Then
\begin{equation}
 \lambda_{\min}(K)+c_{\min}\leq\lambda_{\min}(A)\leq\lambda_{\max}(A)\leq\lambda_{\max}(K)+c_{\max}.
 \label{eq:spectral}
\end{equation}
For conformal coordinates on a rectangle of side lengths $L_1,L_2$ with $N$ interior points per direction and $\varrho=1$, the eigenvalues of $K$ are $\sum_{j=1,2}(4/\Delta_j^2)\sin^2(\pi k_j/(2N+2))$ with $\Delta_j=L_j/(N+1)$.
\end{proposition}
\begin{proof}
For any unit vector $v$, $\langle v,Av\rangle=\langle v,Kv\rangle+\langle v,Cv\rangle$, and both terms lie between the extreme eigenvalues of the respective matrix.
\end{proof}

For the lens of Sec.~\ref{sec:ducts} with $N=32$ and $\Gamma=100$, Eq.~\eqref{eq:spectral} gives $a=2.32\leq\lambda_{\min}=3.07$ and $\lambda_{\max}=1045\leq B=1062$. For the disk it gives $1.07\leq1.83$ and $509\leq552$. These certified values can be used directly in Eq.~\eqref{eq:spectrum}. Weyl's law~\cite{Chavel1984} states that the number of eigenvalues below $\lambda$ grows as $|\Omega|_g\lambda/(4\pi)$, where $|\Omega|_g$ is the area. It is a high-frequency asymptotic statement and does not certify low-lying gaps or a finite dephasing window.

\section{Validation of the discretizations}\label{app:validation}
Table~\ref{tab:valid} validates both discretizations against exact steady solutions. The elliptic duct shows second-order convergence of the velocity and the flow rate. For the lens we use the disk ($\sigma_0=\pi/2$). Its flow rate converges to $0.03\%$ on the $128\times128$ grid, while the pointwise error is limited at about $1\%$ by the truncation at $\tau_{\max}=6$.

\begin{table*}[t]
\centering
\caption{Validation against exact steady solutions. Ellipse: $B_e/A_e=0.5$, quarter domain. Disk: lens discretization with $\sigma_0=\pi/2$, $\tau_{\max}=6$. Errors are relative to the maximum velocity and to the exact flow rate.}\label{tab:valid}
\begin{tabular}{rcccc}
\toprule
& \multicolumn{2}{c}{Ellipse} & \multicolumn{2}{c}{Disk (lens discretization)}\\
\cmidrule(lr){2-3}\cmidrule(lr){4-5}
$N$ & max.\ velocity error & flow-rate error & max.\ velocity error & flow-rate error\\
\midrule
16 & $1.5\times10^{-3}$ & $+2.8\times10^{-3}$ & $8.1\times10^{-3}$ & $-1.1\times10^{-2}$\\
32 & $3.7\times10^{-4}$ & $+6.9\times10^{-4}$ & $6.9\times10^{-3}$ & $-3.0\times10^{-3}$\\
64 & $9.2\times10^{-5}$ & $+1.7\times10^{-4}$ & $8.2\times10^{-3}$ & $-8.6\times10^{-4}$\\
128 & $2.3\times10^{-5}$ & $+4.3\times10^{-5}$ & $9.0\times10^{-3}$ & $-2.9\times10^{-4}$\\
\bottomrule
\end{tabular}
\end{table*}

\bibliographystyle{IEEEtran}
\bibliography{references}
\end{document}